\documentclass[11pt]{article}

\usepackage{amssymb}
\usepackage{amsmath, bbm}
\usepackage{amsthm}
\usepackage{tikz}
\usepackage{mathrsfs}
\usepackage{subfigure}
\usepackage{verbatim}
\usepackage{enumerate}
\usepackage{mathtools,stmaryrd}
\usepackage{hyperref}

\graphicspath{{Figures/}}

     \newcommand{\abs}[1]{\left|#1\right|}

\newcommand{\E}{{\mathbb E}}

\newcommand{\bZ}{{\mathbb Z}}
\newcommand{\Z}{{\mathbb Z}}
\newcommand{\R}{{\mathbb R}}

\numberwithin{equation}{section}
\renewcommand{\P}{\mathbb{P}}
     
\newtheorem{theorem}{Theorem}[section]

\newtheorem{lemma}[theorem]{Lemma}
\newtheorem{corollary}[theorem]{Corollary}
\newtheorem{proposition}[theorem]{Proposition}

\newtheorem*{notation*}{Notation}

\renewenvironment{abstract}
 {\small
  \begin{center}
  \bfseries \abstractname\vspace{-.5em}\vspace{0pt}
  \end{center}
  \list{}{%
    \setlength{\leftmargin}{16.5mm}
    \setlength{\rightmargin}{\leftmargin}%
  }%
  \item\relax}
 {\endlist}

\title{
Dynamics of advantageous mutant spread in \\ spatial death-birth and birth-death Moran models
}

\author{Jasmine Foo$^{1, \ast}$ \and\hspace*{-6pt} Einar Bjarki Gunnarsson$^{2, \ast}$ \and\hspace*{-6pt}  Kevin Leder$^{2}$ \and\hspace*{-6pt}David Sivakoff$^{3}$}
\date{%
    \footnotesize $^1$School of Mathematics, University of Minnesota, Twin Cities, MN 55455, USA. \\[2pt]
    $^2$Department of Industrial and Systems Engineering, University of Minnesota, Twin Cities, MN 55455, USA. \\[2pt]
    $^3$Departments of Statistics and Mathematics, The Ohio State University, OH 43210, USA. \\
        $^\ast$Corresponding authors: jyfoo@umn.edu (Jasmine Foo), gunna042@umn.edu (Einar Bjarki Gunnarsson).
}

\begin{document}

\maketitle

\begin{abstract}
The spread of an advantageous mutation through a population
is of fundamental interest in population genetics.
While the classical Moran model is formulated for a well-mixed population,
it has long been recognized that in real-world applications,
the population usually has an explicit spatial structure which
can significantly 
influence 
the dynamics.
In the context of cancer initiation in epithelial tissue,
several recent works have analyzed the dynamics of advantageous mutant
spread on integer lattices, using the biased voter model
from particle systems theory.
In this spatial version of the Moran model,
individuals first reproduce according to their
fitness and then replace a neighboring individual.
From a biological standpoint, the opposite dynamics,
where individuals first die and are then replaced 
by a neighboring individual according to its fitness, are equally relevant.
Here, we investigate this death-birth analogue of the biased voter model.
We construct the process mathematically, derive the associated dual process,
establish bounds on the survival probability of a single mutant,
and prove that the process has an asymptotic shape.
We also briefly discuss alternative birth-death and death-birth dynamics,
depending on how the mutant fitness advantage affects the dynamics.
We show that birth-death and death-birth formulations
of the biased voter model are equivalent
when fitness affects the former event of each update of the model,
whereas the birth-death model is fundamentally different from the death-birth model when fitness affects the latter event.
\end{abstract}

\section{Introduction}

The spread of an advantageous mutation through a population is of fundamental interest in population genetics.
In the classical Moran model, a population of size $N$ has two types of individuals, normal and mutant.
A normal individual has fitness 1 and a mutant has fitness $\lambda = 1+\beta$ with $\beta>0$.
Each individual $x$ is replaced at exponential rate 1 by the offspring of an individual chosen from the population (including $x$ itself) with probability proportional to its fitness \cite{moran1958random,durrett2008probability}.
For example, if there are $j$ mutants in the population at the time individual $x$ is replaced, the new individual is a mutant with probability $j\lambda/(j\lambda+(N-j))$.
When a single mutant is introduced to an otherwise normal population, a key evolutionary question concerns the probability that the mutation takes over the population.
For the Moran model, this mutant survival probability or fixation probability is easily calculated to be $\rho = (1-1/\lambda)/(1-(1/\lambda)^N)$.

In the Moran model, the population is assumed to be well-mixed, meaning that there is no sense of a spatial structure or distances between individuals.
The stepping stone model of Kimura and Weiss represents an early attempt to model the effect of a geographical structure on the evolution of genetic diversity \cite{kimura1964stepping}.
In works by Maruyama \cite{maruyama1970fixation,maruyama74} and Slatkin \cite{slatkin1981fixation}, the fixation probability $\rho$ of the Moran model was shown to generalize to spatially subdivided populations, under certain assumptions on migration between colonies.
Lieberman et al.~\cite{lieberman2005evolutionary} observed more recently that if individuals are arranged on the nodes of a graph, individuals are chosen for reproduction proportional to their fitness, and an individual at node $i$ places its offspring at node $j$ with probability $w_{ij}$, the fixation probability of a mutant is equal to $\rho$
if and only if 
the weight matrix $W = (w_{ij})$ is doubly stochastic.
For an undirected and unweighted graph, 
this condition is equivalent to all nodes having the same degree \cite{broom2008analysis}.
On degree-heterogeneous or directed graphs, the fixation probability of an advantageous mutant can be either larger or smaller than $\rho$, in which case the graph can be thought of as amplifying or suppressing selection of the mutant, respectively \cite{lieberman2005evolutionary,antal2006evolutionary,sood2008voter,masuda2009evolutionary,hindersin2015most}.

In addition to capturing the effect of geographical dispersion on the population genetics of organisms, a spatial version of the Moran model is relevant to the dynamics of cancer initiation in the human body.
Most cancers originate from epithelial tissue, which covers the inner and outer surfaces of organs and blood vessels in the body \cite{NIH}.
Epithelial tissue is made up of cells which are compactly packed and arranged into one or more structured layers.
In the 1970s, Williams and Bjerknes \cite{williams1972stochastic} suggested a simple model of the spread of cancer cells through the basal layer of epithelial tissue.
In the model, cells of two types, normal and cancerous, are arranged on a two-dimensional lattice.
Normal cells divide at exponential rate 1, and cancer cells divide at rate $\lambda=1+\beta$ with $\beta>0$.
Upon cell division, one daughter cell stays put, and the other replaces a neighboring cell chosen uniformly at random.
This model is known within the field of interacting particle systems as the {\em biased voter model}, where it is studied  on the integer lattice $\mathbb{Z}^d$ of arbitrary dimension $d \geq 1$ \cite{BramsonGriffeath81}.

Cancer generally arises through a multi-stage process of genetic mutations \cite{armitage1954age,armitage1957two,knudson1971mutation,knudson2001}.
As cells acquire more mutations, they become progressively more malignant, which can enable premalignant cells to expand into clones or ``fields'' that are predisposed to becoming cancerous.
The notion that cancer arises through a succession of mutation and premalignant field expansion is referred to as ``field cancerization'' or ``the cancer field effect''.
In recent years, several works have applied the biased voter model to study cancer initiation and the cancer field effect in epithelial tissue.
These works utilize a ``shape theorem'' due to Bramson and Griffeath, which states that conditional on nonextinction, an advantageous mutant clone eventually assumes a convex shape, whose diameter grows linearly in time \cite{BramsonGriffeath81,BraGri80}.
We refer to \cite{DurMose2015,durrett2016spatial,foo2020mutation} for investigations of the time of cancer initiation under a two-step model of cancer, assuming a small mutant selective advantage (weak selection).
We refer to \cite{FLR2014} for a study of the cancer field effect, and to \cite{foo2020spread} for a study of premalignant field expansion on stacked two-dimensional lattices, which represent multilayered epithelial tissue.

In real epithelial tissue, the coordination between cell proliferation and cell death to maintain tissue homeostasis remains poorly understood \cite{brock2019stem}.
This coordination is difficult to study {\em in situ} due to the rapid clearance of apoptotic cells from epithelial tissue \cite{brock2019stem,fuchs2015live}.
In a recent work, Brock et al.~studied the zebrafish epidermis, chosen for its peripheral location and optical clarity \cite{brock2019stem}.
They induced damage in a subset of epithelial stem cells, and they observed that dying cells generated Wnt8a-containing apoptotic bodies.
The apoptotic bodies were engulfed by neighboring stem cells, which activated Wnt signaling and stimulated proliferation within these cells.
Brock et al.~concluded that ingestion of apoptotic bodies represents a mechanism for maintaining homeostasis in epithelial tissue.
Apoptosis-induced proliferation has also been observed in other organisms in response to cellular stress or tissue damage \cite{fuchs2015live}.
In the biased voter model, it is assumed that cell division precedes cell death (birth-death model).
From a biological standpoint, it is also plausible that balance is maintained through cells first dying and then inducing neighboring cells to divide in their place (death-birth model).

Death-birth spatial models have already received interest in the more general context of evolution on a graph, for example regarding mutant fixation.
Formulations of these models vary depending on whether the mutant fitness advantage affects the death event or the birth event, each of which is biologically relevant \cite{kaveh2015duality}.
In the former case, an individual is first selected to die with probability inversely proportional to its fitness, and a random neighboring individual then reproduces in its place (${\rm D}^{\rm f}{\rm B}$ model).
In the latter case, a random individual is selected to die, and a neighbor is chosen to reproduce with probability proportional to its fitness (${\rm DB}^{\rm f}$ model).
Antal et al.~showed that on degree-heterogeneous graphs, the fixation probability of a single mutant depends on the degree of the node at which the mutant appears, and this dependence is different for ${\rm B}^{\rm f}{\rm D}$ and ${\rm D}^{\rm f}{\rm B}$ dynamics \cite{antal2006evolutionary,sood2008voter}.
Hindersin and Traulsen showed that most small undirected random graphs amplify selection for ${\rm B}^{\rm f}{\rm D}$ dynamics and suppress selection for ${\rm D}{\rm B}^{\rm f}$ dynamics \cite{hindersin2015most}.
Kaveh et al.~\cite{kaveh2015duality}, motivated by Komarova \cite{komarova2006spatial}, investigated the fixation probabilities of more general ${\rm B}^{\rm f}{\rm D}^{\rm f}$ and ${\rm D}^{\rm f}{\rm B}^{\rm f}$ models, where fitness affects one or both of the updating events.
They showed that on a circle, ${\rm BD}^{\rm f}$ and ${\rm DB}^{\rm f}$ models deviate from the Moran fixation probability $\rho$.
They also derived approximations of the fixation probabilities of birth-death and death-birth models on regular two-dimensional lattices.

In this work, we investigate a death-birth analogue of the biased voter model on $\mathbb{Z}^d$,
where fitness affects the birth event (${\rm DB}^{\rm f}$ model).
Our investigation involves
(i) constructing the process mathematically, using tools from the theory of interacting particle systems,
(ii) deriving the associated dual process, which traces the lineages of particles backwards in time, 
(iii) establishing bounds of the survival probability for a single mutant, which confirm that an advantageous mutant always has a positive probability of surviving,
(iv) extending the Bramson-Griffeath shape theorem to death-birth dynamics.
We also briefly discuss other birth-death and death-birth analogues of the biased voter model,
depending on when fitness is applied.
We find that whereas models with ${\rm B}^{\rm f}{\rm D}$ and ${\rm D}^{\rm f}{\rm B}$ dynamics are mathematically equivalent, the ${\rm BD}^{\rm f}$ model has fundamentally different properties from the ${\rm DB}^{\rm f}$ model.
In addition, we show that due to the equivalence of the ${\rm D}^{\rm f}{\rm B}$ model to the biased voter model (${\rm B}^{\rm f}{\rm D}$ model),
our analysis for the ${\rm DB}^{\rm f}$ model extends relatively easily to more general ${\rm D}^{\rm f}{\rm B}^{\rm f}$ dynamics.

The rest of the paper is organized as follows.
In Section \ref{sec:bvm}, we briefly review key properties of the biased voter model on $\mathbb{Z}^d$.
In Section \ref{sec:dbmodeldef}, we formulate our death-birth model on $\mathbb{Z}^d$.
In Section \ref{sec:onedanalysis}, we present complete analysis of this model on $\mathbb{Z}$ ($d=1$).
In Section \ref{sec:gencasegraphical}, we develop a graphical construction of the death-birth model on $\mathbb{Z}^d$, $d>1$, and derive its dual process.
In Section \ref{sec:DBsurvprob}, we derive bounds on the survival probability of the model, starting from a single mutant.
In Section \ref{sec:shapethm}, we extend the Bramson-Griffeath shape theorem to the death-birth model.
In Section \ref{sec:altmodels}, we briefly discuss alternative birth-death and death-birth dynamics on $\mathbb{Z}^d$.
In Section \ref{sec:openproblems}, we present open problems related to the behavior of the death-birth model in the limit of weak selection.
In Section \ref{sec:proofs}, we give proofs that were deferred in earlier sections.

\section{Biased voter model} \label{sec:bvm}

We begin with a brief review of the biased voter model on $\mathbb{Z}^d$ and its main properties.
This discussion will both introduce the basic tools from particle systems theory used in our analysis of the death-birth model and draw a useful basis for comparison between the two models.

\subsection{Model description} \label{sec:bvmdescr}

In the biased voter model (${\rm B}^{\rm f}{\rm D}$ model), two types of particles, type-0 and type-1, are situated on $\mathbb{Z}^d$.
Type-0 particles split into two particles at exponential rate 1, while type-1 particles split at rate $\lambda = 1+\beta$ with $\beta>0$.
A particle splitting can also be thought of as the particle giving birth to a new particle of the same type.
When a particle gives birth, a neighboring particle is chosen to die uniformly at random, and it gets replaced by the new particle.

Say that a particle at $x \in \mathbb{Z}^d$ has $i$ type-0 neighbors and $j$ type-1 neighbors.
If the particle is type-0, it gets replaced by a neighboring type-1 particle at rate $j\lambda/(2d)$.
If it is type-1, it gets replaced by a neighboring type-0 particle at rate $i/(2d)$.
We can think of this as type-0 sites (resp.~type-1 sites) switching to type-1 (resp.~type-0) at rate $j\lambda /(2d)$ (resp.~$i/(2d)$). 

\subsection{Graphical construction} \label{sec:bvmgraphical}

Let ${\cal S}$ denote the set of subsets of $\Z^d$ and let $\overline{\cal S}$ denote the set of finite subsets of $\Z^d$.
Let $\chi_t^A$ denote the set of sites occupied by type-1 particles at time $t$,
given the initial condition $\chi_0^A=A$ with $A \in \overline{\cal S}$.
It is common to think of the type-0 sites as empty and refer to $\chi_t^A$ simply as the set of ``occupied sites''.
The biased voter model admits a simple graphical construction which enables us to define the entire system $\{(\chi_t^A)_{t \geq 0}: A \in \overline{\cal S}\}$ on a common probability space.
The construction is as follows:
For a given site $x \in \Z^d$ and each nearest neighbor $y$ of $x$, draw an arrow from $y$ to $x$ with a ``$\delta$'' at the arrowhead at rate $1/(2d)$, and draw a regular arrow (without the $\delta$) at rate $\beta/(2d)$.
A $\delta$-arrow from $y$ to $x$ kills the particle at $x$ and replaces it with an offspring of the particle at $y$.
A regular arrow from $y$ to $x$ kills the particle at $x$ only if it is type-0.
In Figure \ref{fig:realizationbvm}a, we show the types of arrows that can be drawn to a given site $x$ for the biased voter model on $\mathbb{Z}$.
In Figure \ref{fig:realizationbvm}b, we show the possible states of $x \in \Z$ after an arrow is drawn from $y$ to $x$, depending on the states of $x$ and $y$ immediately before the arrow is drawn.
It is useful to picture type-1 sites as being wet with fluid, the $\delta$'s as dams that block fluid from flowing further, and that arrows allow fluid to flow between neighboring sites in the direction indicated.
In Figure \ref{fig:realizationbvm}c, we show a possible graphical realization of the biased voter model on $\Z$ under this interpretation.
For more information on this graphical construction, we refer to Section 3 of \cite{durrett1988lecture} or Appendix A of \cite{durrett2016spatial}.

\begin{figure}[t]
    \centering
    \includegraphics[scale=0.4]{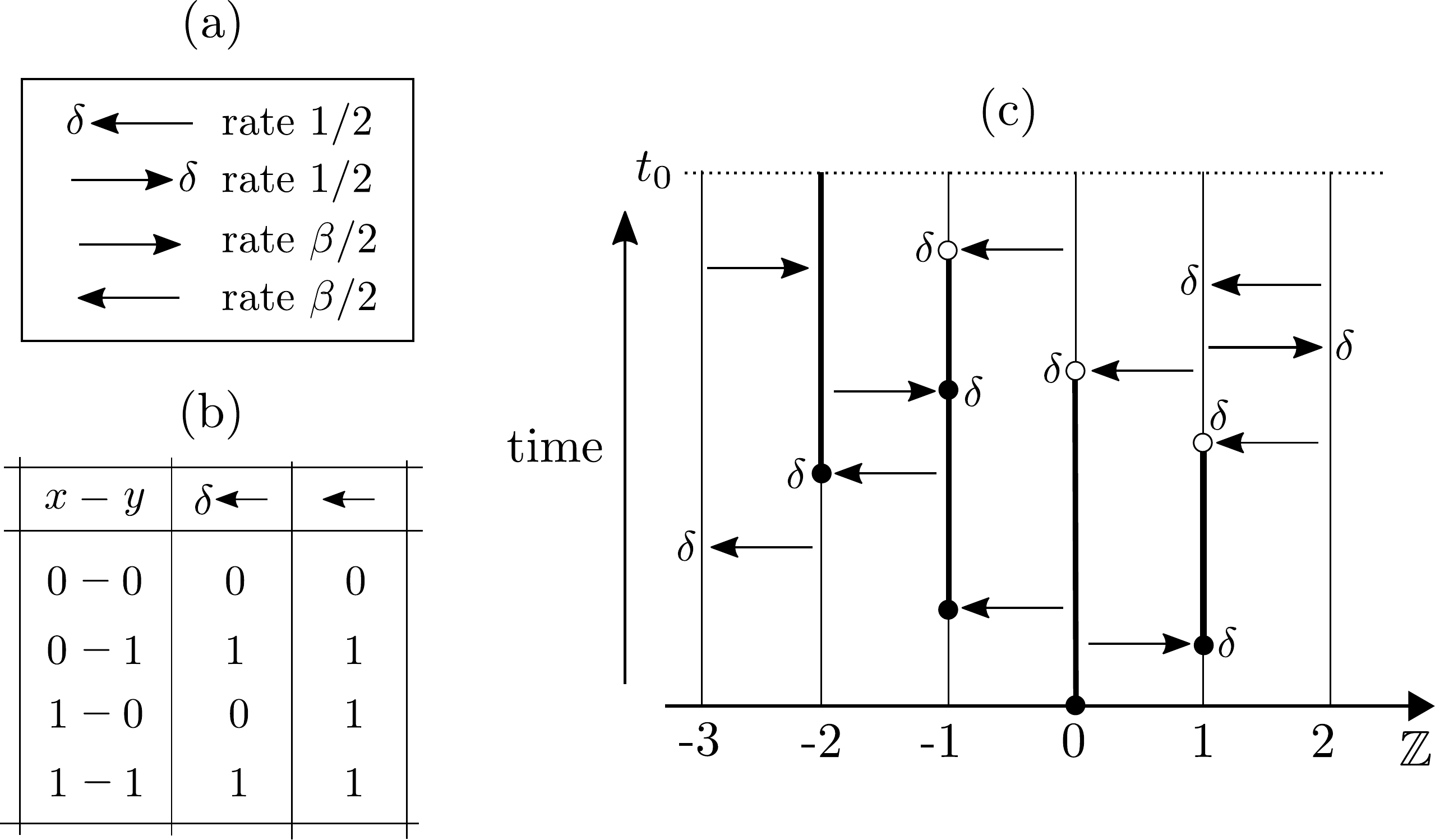}
    \caption{
    The basic tools used in the graphical construction of the biased voter model $(\chi_t)_{t\geq 0}$ on $\Z$.
    {\bf (a)} Types of arrows used in the graphical construction. For a given site $x \in \Z$, a $\delta$-arrow is drawn from each of its two neighbors at rate $1/2$ each, and a regular arrow is drawn at rate $\beta/2$.
    {\bf (b)} State of $x \in \mathbb{Z}$ after an arrow has been drawn from $y$ to $x$, depending on the state of $x$ and $y$ immediately before (leftmost column), and the type of arrow drawn (top row).
    {\bf (c)} Graphical realization of $(\chi_t^0)_{0 \leq t \leq t_0}$ on $\Z$ for some $t_0>0$. Site 0 is wet with fluid at time 0, and the fluid flows along the arrows drawn until it meets a $\delta$-dam. At time $t_0$, only the site $-2$ is wet, so $\chi_{t_0}^0 = \{-2\}$ for this particular realization.}
    \label{fig:realizationbvm}
\end{figure}

\subsection{Dual process} \label{sec:bvmdual}

Associated with the graphical construction is a dual process which traces the lineages of type-1 particles backwards in time. 
In Figure \ref{fig:realizationbvm_dual}, starting with type-1 particles at sites $-2$ and $1$ at time $t_0$ and moving backwards in time, we follow each arrow in the reverse direction.
If a particle encounters a $\delta$-arrow, it jumps to the neighboring site at the other end.
If it encounters a regular arrow, it gives birth to a new particle, which is placed at the neighboring site.
Whenever a particle attempts to occupy an already occupied site, the two particles coalesce.
In this way, we keep track of the possible ancestors of the original particles as we move backwards in time.
For example, for the realization in Figure \ref{fig:realizationbvm_dual}, the ancestors of the type-1 particles at sites $-2$ and $1$ at time $t_0$ must have resided at sites $-2$, $-1$, $0$ or $2$ at time 0.
In the dual process, $\{-2,1\}$ is the initial condition and $\{-2,-1,0,2\}$ is the state of the process at time $t_0$.

The dual process can be more succinctly described as a branching coalescing random walk (BCRW), where particles jump at rate 1 to a randomly chosen neighboring site, and they branch at rate $\beta$, placing an offspring at a randomly chosen neighboring site.
Any time two particles meet at the same site, they coalesce into a single particle.
For $B \in {\cal S}$, let $\eta_t^B$ denote the set of occupied sites in the dual process at time $t$, given initial condition $\eta_0^B = B$.
The relationship between the dual process $(\eta_t^B)_{t \geq 0}$ and the original process $(\chi_t^A)_{t \geq 0}$ is encapsulated in the following duality relation:
\begin{align} \label{eq:dualityrelationbvm}
    \P\big(\chi_t^A \cap B \neq \varnothing\big) = \P\big(\eta_t^B \cap A \neq \varnothing\big), \quad A \in \overline{\cal S}, \; B \in {\cal S}, \; t \geq 0.
\end{align}
This relation can be interpreted as follows: The probability that the forwards-in-time process reaches the set $B$ at time $t$, starting with the set $A$ occupied, is equal to the probability that the ancestry of at least one particle in $B$ at time $t$ can be traced back to the set $A$ at time 0 using the backwards-in-time process.
Both the biased voter model and its dual process are {\em additive} in the sense that for any $A,B \in \overline{\cal S}$ and $C,D \in {\cal S}$,
\begin{align*}
    \chi_t^{A \cup B} = \chi_t^A \cup \chi_t^B \quad\text{and}\quad \eta^{C \cup D}_t = \eta^C_t \cup \eta^D_t, \quad t \geq 0.
\end{align*}

\begin{figure}[t]
    \centering
    \includegraphics[scale=0.4]{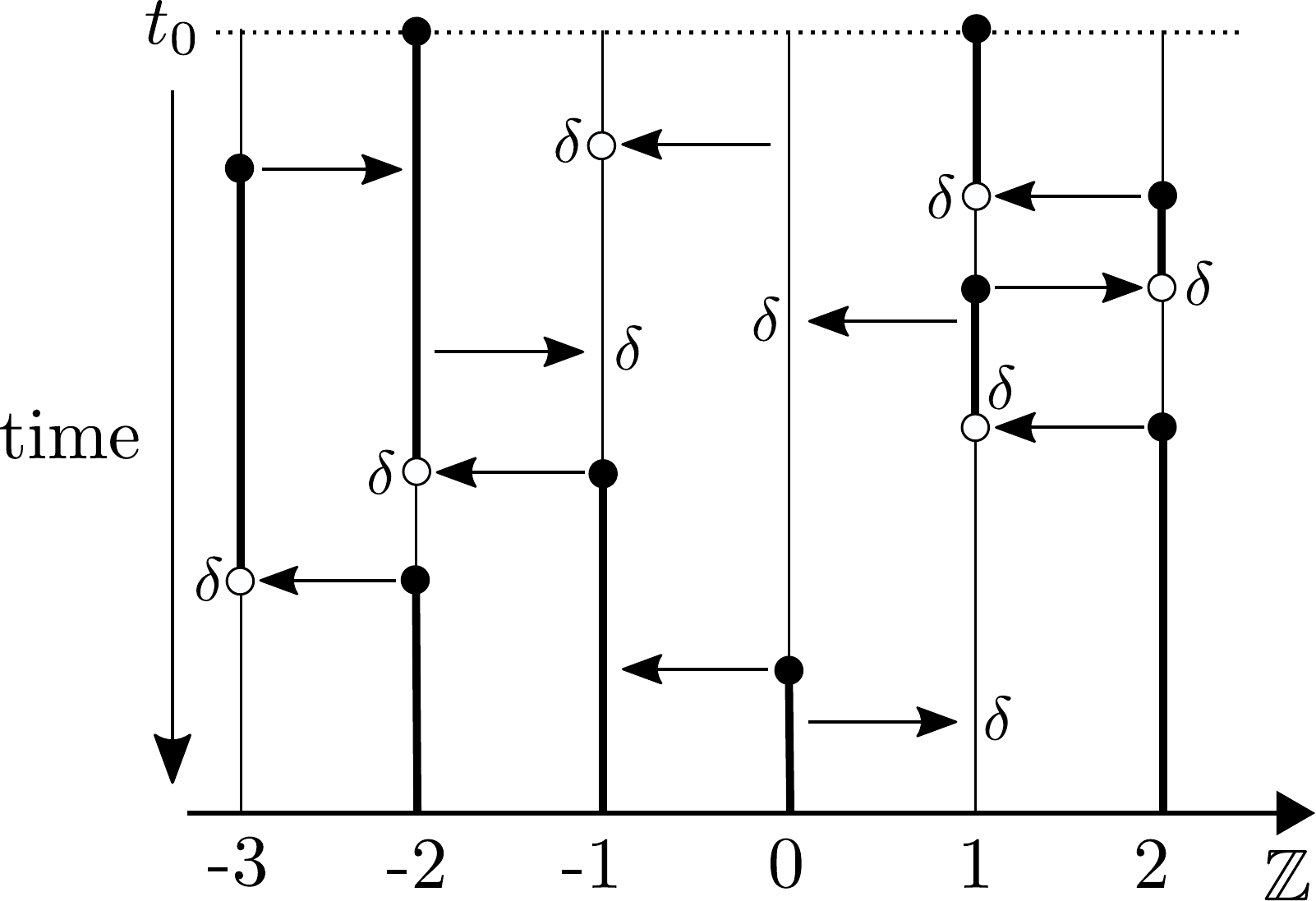}
    \caption{Mechanics of the dual process for biased voter model. If we start with type-1 particles at sites $-2$ and 1 at time $t_0$ and move backwards in time, following each arrow in the reverse direction, we obtain all possible ancestors of these particles at time 0.
    We observe that the ancestors must have resided at sites $-2$, $-1$, $0$ or $2$ at time 0.
    For this particular realization of the dual process, the initial condition is $\{-2,1\}$, and its state at time $t_0$ is $\{-2,-1,0,2\}$.
    }
    \label{fig:realizationbvm_dual}
\end{figure}

\subsection{Survival probability} \label{sec:bvmsurvprob}

Define $\widetilde\tau_\varnothing^A := \inf\{t \geq 0: \chi_t^A = \varnothing\}$, with $\inf \varnothing = \infty$, as the time of extinction of type-1 particles, starting with the set $A$ occupied.
We can compute the probability $\P(\widetilde\tau_\varnothing^A=\infty)$ of type-1 survival by considering the discrete-time jump process $(R_n^A)_{n \geq 0}$ embedded in $(|\chi_t^A|)_{t \geq 0}$, where $|\cdot|$ denotes cardinality.
The jump process is defined by setting $R_0^A := |A|$, and if $\sigma_n$ is the time of the $n$-th jump of $(|\chi_t^A|)_{t \geq 0}$, then $R_n^A := |\chi_{\sigma_n}^A|$.
Note that $(|\chi_t^A|)_{t \geq 0}$ jumps when a type-1 particle gives birth and overtakes a neighboring type-0 particle or vice versa.
For each edge between a type-0 and type-1 particle, the type-1 particle overtakes the type-0 particle at rate $\lambda/(2d)$, and the type-0 particle overtakes the type-1 particle at rate $1/(2d)$.
The type-1 particle wins with probability $\lambda/(1+\lambda)$.
It follows that $(R_n^A)_{n \geq 0}$ is a simple random walk with absorption at 0 and a uniform positive drift $\Delta := (\lambda-1)/(1+\lambda) = \beta/(2+\beta)$ on $\{1,2,\ldots\}$.
Thus, by the gambler's ruin formula and the strong Markov property,
\begin{align} \label{eq:extinctionprobabilitybvm}
    \P(\widetilde\tau_\varnothing^A<\infty) = (1/\lambda)^{|A|}.
\end{align}
The corresponding survival probability is
\begin{align*}
    \P(\widetilde\tau_\varnothing^A=\infty) = 1-(1/\lambda)^{|A|}.
\end{align*}
In particular, the survival probability starting with the origin occupied is
\begin{align} \label{eq:survprobbvm}
    \P(\widetilde\tau_\varnothing^0=\infty) = (\lambda-1)/\lambda = \beta/(1+\beta).
\end{align}
Note that
\begin{align} \label{eq:ordersurvprobbvm}
    \P(\widetilde\tau_\varnothing^0=\infty) \sim \begin{cases} \beta, & \beta \to 0, \\ 1, & \beta \to \infty. \end{cases}
\end{align}
Here, $f(x) \sim g(x)$ as $x \to a$, with $a \in \mathbb{R} \cup \{-\infty,\infty\}$, means that $f(x)/g(x) \to 1$ as $x \to a$.
In other words, in the limit of weak selection ($\beta \to 0$), the survival probability $\P(\widetilde\tau_\varnothing^0=\infty)$ is of order $\beta$, while in the limit of strong selection ($\beta \to \infty$), it converges to 1.
The limit of weak selection is relevant in the context of cancer initiation.
For example, in \cite{bozic2010accumulation}, Bozic et al.~show that data on multiple cancer types are consistent with a multi-stage model of carcinogenesis with a selective advantage of $\beta = 0.004$ per mutational step.

\subsection{Asymptotic shape} \label{sec:bvmshape}

On the event $\{\widetilde\tau_\varnothing^0=\infty\}$, the Bramson-Griffeath shape theorem shows that the set of type-1 particles eventually assumes a convex shape, whose diameter grows linearly in time \cite{BramsonGriffeath81,BraGri80}.
More precisely, there exists a convex subset $\widetilde D$ of $\R^d$ so that for each $\varepsilon>0$,
\begin{align} \label{eq:shapethmbvm}
\P\big(\exists t_*<\infty: (1-\varepsilon)t \widetilde D \cap \mathbb{Z}^d \subseteq \chi_t^0 \subseteq (1+\varepsilon)t \widetilde D,\; t\geq t_* \big| \widetilde\tau_\varnothing^0=\infty\big)=1.
\end{align}
This concludes our review of the biased voter model.

\section{Death-birth model formulation} \label{sec:dbmodeldef}

We are now ready to formulate the death-birth model (${\rm DB}^{\rm f}$ model) on $\mathbb{Z}^d$.
In this model, all particles die at rate 1, 
and the fitness advantage of type-1 particles is incorporated into the subsequent birth event.
Type-0 particles have proliferative fitness 1, and type-1 particles have fitness $\lambda = 1+\beta$ with $\beta>0$.
When a particle dies at $x \in \mathbb{Z}^d$, a neighboring particle is selected to give birth and replace the particle at $x$, with probability proportional to its fitness.

Say that a particle at $x \in \mathbb{Z}^d$ has $i$ type-0  neighbors and $j$ type-1 neighbors.
If $x$ is type-0, the rate of switching to type-1 is $j\lambda/(i+j\lambda)$.
If $x$ is type-1, the rate of switching to type-0 is $i/(i+j\lambda)$.
Recall from Section \ref{sec:bvmdescr} that the corresponding switching rates for the biased voter model are $j\lambda/(2d)$ and $i/(2d)$, respectively.
Note that for both models, the rate of switching for a given site is determined by the number of neighbors of each type.
However, for the death-birth model, the switching rates are no longer linear in $i$, $j$ and $\lambda$.

Let $\xi_t^A$ denote the set of occupied sites at time $t$ under this model, given the initial condition $\xi_0^A = A$ with $A \in \overline{\cal S}$.
Define $\tau_\varnothing^A := \inf \{ t \geq 0: \xi_t^A = \varnothing\}$ 
as the time of extinction of type-1 particles, starting with the set $A$ occupied.
Finally, let $(S_n^A)_{n \geq 0}$ be the discrete-time jump process embedded in $(|\xi_t^A|)_{t \geq 0}$, defined analogously to Section \ref{sec:bvmsurvprob}.

\section{Analysis of death-birth model in one dimension} \label{sec:onedanalysis}

We begin by analyzing the death-birth model in one dimension, that is, on $\Z$.
We develop its graphical construction, determine the associated dual process, compute its survival probability, and determine its asymptotic shape.
The reason we start with the one-dimensional case is threefold.
First, the graphical construction and the dual process are easiest to develop on $\Z$, and the development extends in a straightforward manner to $\Z^d$ for $d>1$.
Second, we can give a complete characterization of the survival probability and asymptotic shape in one dimension, which turns out to be much more difficult for higher dimensions.
Third, the analysis in one dimension is instructive for understanding the complications that arise in higher dimensions.

\subsection{Graphical construction} \label{sec:onedimgraphical}

\begin{figure}[t]
    \centering
    \includegraphics[scale=0.4]{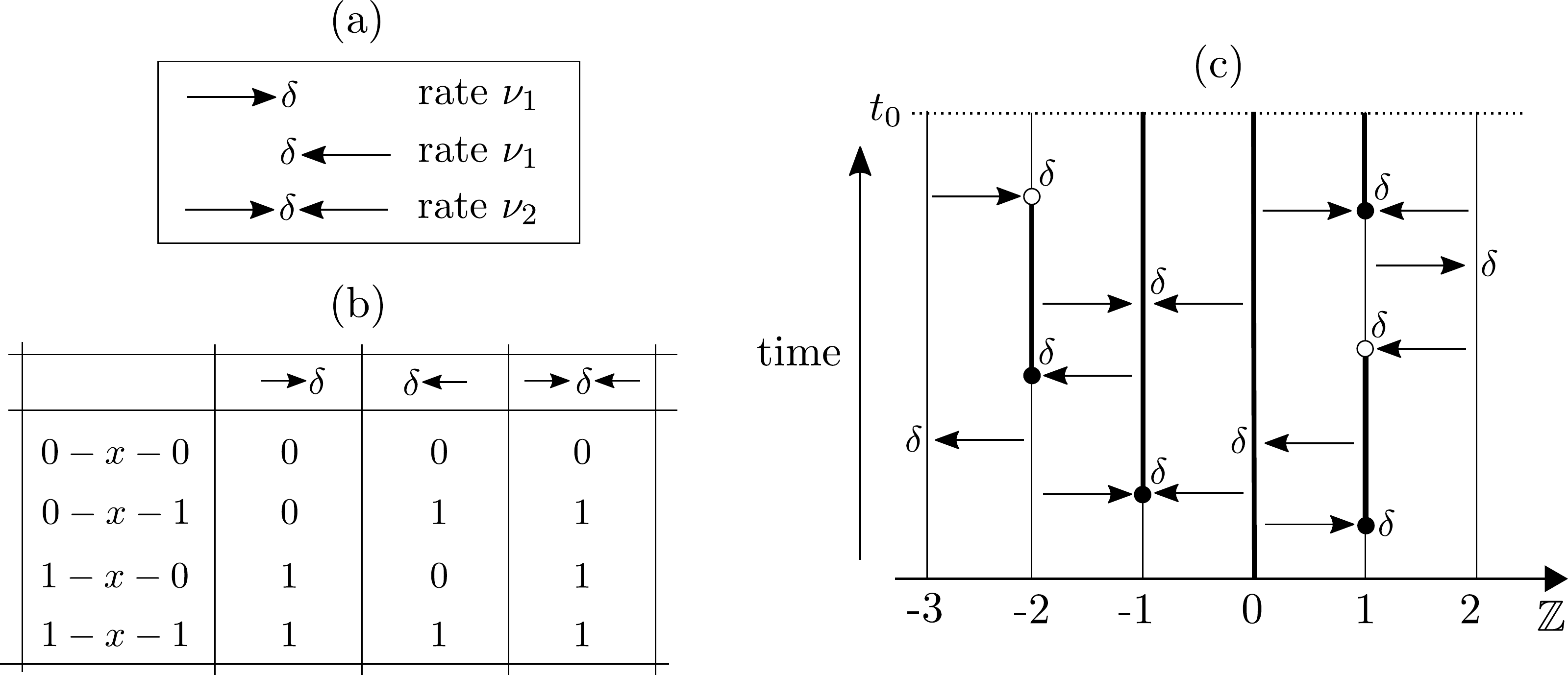}
    \caption{
    The basic tools used in the graphical construction of the death-birth model $(\xi_t)_{t\geq 0}$ on $\Z$.
    {\bf (a)} Types of arrows used in the graphical construction. {\bf (b)} State of $x \in \mathbb{Z}$ after an arrow has been drawn to $x$, depending on the state of its neighbors immediately before (leftmost column), and the type of arrow drawn (top row). {\bf (c)} Graphical realization of $(\xi_t^0)_{t \geq 0}$ on $\mathbb{Z}$ for some $t_0>0$. 
    At time $t_0$, the sites $-1$, $0$ and $1$ are occupied, so $\xi_{t_0}^0 = \{-1,0,1\}$ for this particular realization.}
    \label{fig:realization}
\end{figure}

To construct the system  $\{(\xi_t^A)_{t \geq 0}: A \in \overline{\cal S}\}$ graphically in one dimension, we proceed as follows.
For each $x \in \mathbb{Z}$, draw $\delta$-arrows from $x-1$ to $x$ and from $x+1$ to $x$ at rate $\nu_1$ each, and draw $\delta$-arrows from $x-1$ and $x+1$ to $x$ simultaneously at rate $\nu_2$ (Fig \ref{fig:realization}a).
When an arrow or arrows are drawn to $x$, the $\delta$ kills the particle at $x$, and it gets replaced by a type-1 particle if and only if at least one of the arrows drawn connects $x$ to a type-1 particle.
In Figure \ref{fig:realization}b, we show the possible states of $x \in \mathbb{Z}$ after an arrow is drawn to $x$, depending on the state of its neighbors immediately before the arrow is drawn.
In Figure \ref{fig:realization}c, we show a possible graphical realization of the death-birth model on $\Z$.

We now show that $\nu_1$ and $\nu_2$ can be chosen so that this construction produces the dynamics of the death-birth model on $\mathbb{Z}$.
Note that if a type-0 particle at $x \in \mathbb{Z}$ dies and exactly one of its two neighbors is type-1, $x$ becomes type-1 if and only if a $\delta$-arrow is drawn from this neighbor or both neighbors.
Since under the death-birth model, $x$ should switch from type-0 to type-1 at rate $\lambda/(1+\lambda)$, we must have $\nu_1+\nu_2 = \lambda/(1+\lambda)$.
Additionally, since $x$ should die at total rate 1, we must have $2\nu_1+\nu_2 = 1$.
Together, these expressions yield $\nu_1 = 1/(1+\lambda)$ and $\nu_2 = (\lambda-1)/(1+\lambda)$.
It is then easy to verify that with $\nu_1$ and $\nu_2$ chosen as such, we obtain a simple graphical representation of the death-birth model.

\subsection{Dual process} \label{sec:oneddual}

To obtain the corresponding dual process, we again trace the lineages of type-1 particles backwards in time.
In Figure \ref{fig:realization_dual}, starting with a type-1 particle at site 1 at time $t_0$ and moving backwards in time, we follow each arrow in the reverse direction.
If a particle encounters a single $\delta$-arrow, it jumps to the neighboring site at the other end.
If it encounters a double $\delta$-arrow, it splits into two particles and places one offspring at each neighboring site.
Whenever a particle attempts to occupy a site which is already occupied, it coalesces with the particle there.
For example, for the realization in Figure \ref{fig:realization_dual}, the ancestor of the particle at site 1 at time $t_0$ must have resided at sites $0$ or $2$ at time 0.

The dynamics of the dual process are as follows:
A particle at $x \in \mathbb{Z}$ jumps to $x-1$ at rate $\nu_1$ and to $x+1$ at rate $\nu_1$,
and it splits into two particles which are placed at $x-1$ and $x+1$ at rate $\nu_2$.
Any time two particles meet at the same site, they coalesce into a single particle.
For $B \in {\cal S}$, we let $\zeta_t^B$ denote the set of occupied sites at time $t$ in the dual process, given initial condition $\zeta_0^B = B$. 
We then have the following duality relation between $(\xi_t^A)_{t \geq 0}$ and $(\zeta_t^B)_{t \geq 0}$:
\begin{align} \label{eq:dualitydeathbirth}
\P(\xi_t^A \cap B \neq \varnothing) = \P(\zeta_t^B \cap A \neq \varnothing), \quad A \in \overline{\cal S}, \; B \in {\cal S}, \; t \geq 0.    
\end{align}

\begin{figure}[t]
    \centering
    \includegraphics[scale=0.4]{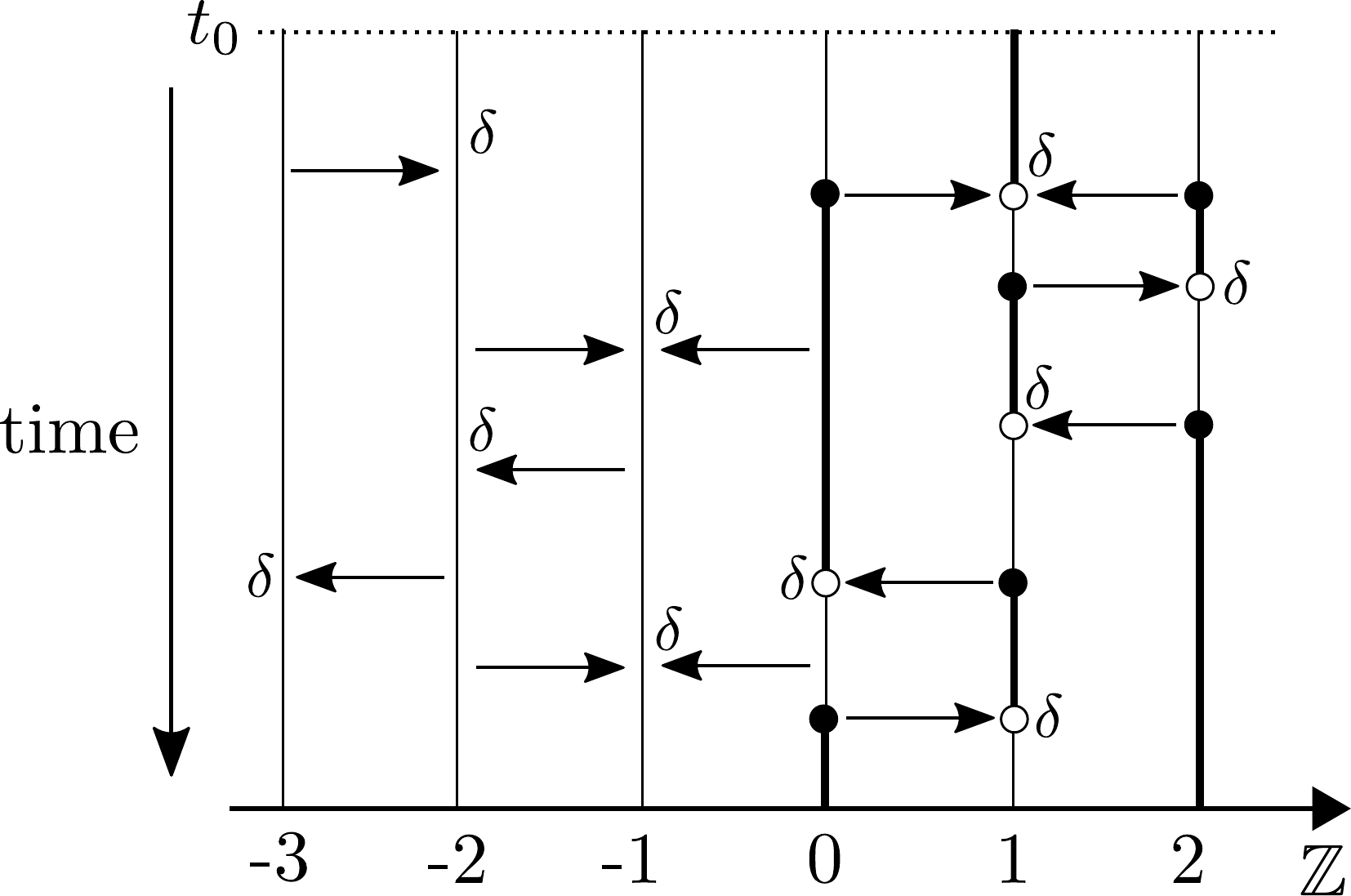}
    \caption{Mechanics of the dual process for the death-birth model. If we start with a type-1 particle at site 1 at time $t_0$ and move backwards in time, following each arrow in the reverse direction, we obtain all possible ancestors of this particle at time 0.
    We observe that the ancestor must have resided at sites 0 or 2 at time 0. For this particular realization of the dual process, the initial condition is $\{1\}$, and its state at time $t_0$ is $\{0,2\}$.}
    \label{fig:realization_dual}
\end{figure}

\subsection{Survival probability} \label{sec:onedsurvprob}

We now analyze the survival probability of $(\xi_t^0)_{t \geq 0}$, starting with the origin occupied.
First, we note that $\xi_t^0$ has the form $\xi_t^0 = \llbracket l_t,r_t\rrbracket$ as long as it survives, and that changes to $\xi_t^0$ occur only at its boundary.
When $l_t = r_t =: x$, then for each neighbor $y \in \{x-1,x+1\}$, $y$ dies and is replaced by an offspring of $x$ at rate $\lambda/(1+\lambda)$, and $x$ dies and is replaced by an offspring of $y$ at rate $1/2$.
The probability that $x$ wins, that is, $x$ replaces $y$ before the opposite occurs, is $2\lambda/(1+3\lambda)$.
When $l_t < r_t$, $r_t+1$ dies and is replaced by an offspring of $r_t$ at rate $\lambda/(1+\lambda)$, and $r_t$ dies and is replaced by an offspring of $r_t+1$ at rate $1/(1+\lambda)$.
The probability that $r_t$ wins is $\lambda/(1+\lambda)$.
The same analysis applies to $l_t$ and $l_t-1$.
It follows that the embedded process $(S_n^0)_{n \geq 0}$ is a simple random walk with absorption at 0, with a positive drift of $(\lambda-1)/(1+3\lambda)$ when $S_n^0=1$ and $(\lambda-1)/(1+\lambda)$ when $S_n^0 \geq 2$.
Recall that the corresponding process for the biased voter model has a uniform drift of $\Delta := (\lambda-1)/(1+\lambda)$ on $\{1,2,\ldots\}$ (Section \ref{sec:bvmsurvprob}).
Since $(\lambda-1)/(1+3\lambda)<\Delta$, it is immediately clear that the survival probability for $(\xi_t^0)_{t \geq 0}$ will be smaller than $\P(\widetilde\tau_\varnothing^0 = \infty) = \beta/(1+\beta)$ for the biased voter model.
To calculate its value, we can use a generalized version of the gambler's ruin formula (Theorem~5.3.11 of~\cite{MR3930614}) to obtain
\begin{align} \label{eq:survprobsol}
    \textstyle \P(\tau_\varnothing^0 = \infty) 
    = {2(\lambda-1)}/({3\lambda-1}) = (2\beta)/(2+3\beta).
\end{align}
This result has appeared in expression (48) of Komarova \cite{komarova2006spatial} and expression (5.3) of Kaveh et al.~\cite{kaveh2015duality}.
Note that
\begin{align*}
    \P(\tau_\varnothing^0=\infty) \sim \begin{cases} \beta, & \beta \to 0, \\ 2/3, & \beta \to \infty. \end{cases}
\end{align*}
In the limit of weak selection $(\beta \to 0)$, the survival probability is of order $\beta$, which is the same as for the biased voter model, see \eqref{eq:ordersurvprobbvm}. 
In the limit of strong selection ($\beta \to \infty$), the survival probability converges to $2/3$, whereas it converges to 1 for the biased voter model.
This reflects the fact that in the death-birth model, particle death precedes particle birth, and the fitness advantage of type-1 particles does not manifest until the latter event.
When the death-birth model is started by a single type-1 particle,
this particle dies as frequently as its type-0 neighbors, which implies that there is always at least 1/3 probability of extinction, independently of $\beta$.
In the limit of strong selection, the survival probability becomes $2/3$.

\subsection{Asymptotic shape} \label{sec:onedshape}

We conclude by deriving the asymptotic shape of the death-birth process on $\Z$ on the event $\{\tau_\varnothing^0=\infty\}$.
In Section \ref{sec:onedsurvprob}, we noted that $\xi_t^0$ has the form $\xi_t^0 = \llbracket l_t,r_t\rrbracket$ as long as it survives, which we used to compute the probability $\P(\tau_\varnothing^0 = \infty) = 2(\lambda-1)/(3\lambda-1)$.
More generally, if $A$ is a contiguous set of sites in $\Z$, then by the gambler's ruin formula and the strong Markov property,
\begin{align} \label{eq:extinctionprobabilityoned}
   \textstyle \P(\tau_\varnothing^A<\infty) = \frac{\lambda+1}{3\lambda-1} \cdot \big(\frac1\lambda\big)^{|A|-1}.
\end{align}
Note that \eqref{eq:extinctionprobabilityoned} is analogous to \eqref{eq:extinctionprobabilitybvm} for the biased voter model, except the latter expression holds for any subset $A$ of $\mathbb{Z}^d$ with $d \geq 1$.
Also recall that whenever $|\xi_t^0|=k$ with $k \geq 2$, the right edge $r_t$ increases at rate $\lambda/(1+\lambda)$ and decreases at rate $1/(1+\lambda)$.
Then, conditional on $\{\tau_\varnothing^0=\infty\}$ and $|\xi_t^0|=k$ with $k \geq 2$, it is straightforward to verify that
\begin{align*}
    r_t \to \begin{cases} r_t+1 & \text{at rate $\frac{\lambda}{1+\lambda} \cdot \big(1-\frac{\lambda+1}{3\lambda-1} (\frac1\lambda)^{k}\big) \cdot \big(1-\frac{\lambda+1}{3\lambda-1} (\frac1\lambda)^{k-1}\big)^{-1}$}, \\ r_t-1 & \text{at rate $\frac1{1+\lambda} \cdot \big(1-\frac{\lambda+1}{3\lambda-1} (\frac1\lambda)^{k-2}\big) \cdot \big(1-\frac{\lambda+1}{3\lambda-1} (\frac1\lambda)^{k-1}\big)^{-1}$}. \end{cases}
\end{align*}
Conditional on $\{\tau_\varnothing^0=\infty\}$ and $|\xi_t^0|=k$ with $k \geq 2$, the net rate at which $r_t$ increases is
\begin{align*}
    \textstyle \frac{\lambda-1}{1+\lambda} \cdot \big(1+\frac{\lambda+1}{3\lambda-1} (\frac1\lambda)^{k-1}\big) \cdot \big(1-\frac{\lambda+1}{3\lambda-1} (\frac1\lambda)^{k-1}\big)^{-1},
\end{align*}
which decreases to $(\lambda-1)/(1+\lambda) = \beta/(2+\beta)$ as $k \to \infty$.
By analyzing $l_t$ similarly, we can show that if $c_1(\beta) := \beta/(2+\beta)$ and $D := [-c_1(\beta),c_1(\beta)]$, then for any $\varepsilon>0$,
\begin{align} \label{eq:shapethemoned}
    \P\big(\exists t_\ast<\infty: (1-\varepsilon)tD \cap \Z \subseteq \xi_t^0 \subseteq (1+\varepsilon)tD, \; t \geq t_\ast \big| \tau_\varnothing^0 = \infty\big) = 1.
\end{align}

\section{Graphical construction for death-birth model on $\mathbb{Z}^d$, $d>1$} \label{sec:gencasegraphical}

In this section, we describe the graphical construction for the system $\{(\xi_t^A)_{t \geq 0}: A \in \overline{\cal S}\}$ on $\mathbb{Z}^d$ with $d>1$ and derive the associated dual process.

\subsection{Graphical construction} \label{sec:gencasegraphicalsubsec}

To construct the death-birth model graphically on $\Z^d$, $d>1$, we follow the construction for the one-dimensional case (Section \ref{sec:onedimgraphical}).
Let ${\cal N}(x)$ be the set of neighbors of $x \in \mathbb{Z}^d$.
For each subset $S \subseteq {\cal N}(x)$ with $|S| = j$, draw $\delta$-arrows from all sites in $S$ to $x$ simultaneously at rate $\nu_j$.
Recall that the $\delta$ kills the particle at $x$, and that $x$ is replaced by a type-1 particle if and only if at least one of the arrows connects $x$ to a type-1 site.

To complete the graphical construction, we need to show that the rates $\nu_1,\ldots,\nu_{2d}$ can be selected so that these dynamics give rise to the death-birth model.
Note first that if a type-1 particle at $x$ dies and it has exactly $i$ type-0 neighbors, $x$ becomes type-0 if and only if all arrows drawn originate from type-0 neighbors.
Since under the death-birth model, $x$ should switch from type-1 to type-0 at rate $i/(i+(2d-i)\lambda)$, we must have
\begin{align} \label{eq:linearsystem}
\textstyle 
\frac{i}{i+(2d-i)\lambda} = \sum_{k=1}^i \binom{i}k \nu_k, \qquad 1 \leq i \leq 2d. 
\end{align}
This system has the unique solution
\begin{align} \label{eq:linearsystemsolution}
\textstyle
\nu_j = \frac{j! (\lambda-1)^{j-1}}{\prod_{k=1}^j (k+(2d-k)\lambda)}, \quad 1 \leq j \leq 2d,
\end{align}
as is verified in Section \ref{app:linearsystem}.
With $\nu_1,\ldots,\nu_{2d}$ chosen as such, we obtain a graphical representation of the death-birth model on $\Z^d$ with $d>1$.

\subsection{Dual process}

As for the one-dimensional case (Section \ref{sec:oneddual}), the dual process is straightforward to describe:
For each $S \subseteq {\cal N}(x)$  with $|S| = j$, a particle at $x$ splits into $j$ particles at rate $\nu_j$, and places its offspring at the sites in $S$.
Any time two particles meet at the same site, the coalesce into a single particle.
As in Section \ref{sec:oneddual}, we let $\zeta_t^B$ denote the set of occupied sites at time $t$ in the dual process, given initial condition $\zeta_0^B = B$, and we obtain the same duality relation \eqref{eq:dualitydeathbirth}.
We also note that both the death-birth model and its dual are additive in the sense that for any $A,B \in \overline{\cal S}$ and $C,D \in {\cal S}$,
\begin{align*}
    \xi_t^{A \cup B} = \xi_t^A \cup \xi_t^B \quad\text{and}\quad \zeta_t^{C \cup D} = \zeta_t^C \cup \zeta_t^D, \quad t \geq 0.
\end{align*}

\section{Survival probability for death-birth model on $\mathbb{Z}^d$, $d>1$} \label{sec:DBsurvprob}

In this section, we analyze the survival probability of the death-birth model $(\xi_t^0)_{t \geq 0}$ on $\mathbb{Z}^d$ with $d>1$, starting with the origin occupied.
We derive lower and upper bounds, which confirm in particular that the survival probability is positive for any fitness advantage $\beta > 0$.

\subsection{Biases of 0--1 edges} \label{sec:biashighdim}

In one dimension, we computed the survival probability for $(\xi_t^0)_{t \geq 0}$ by analyzing the competition along each 0--1 edge between a type-0 and type-1 particle (Section \ref{sec:onedsurvprob}).
Our analysis was simplified by the fact that on $\mathbb{Z}$, $\xi_t^0$ is a contiguous set of sites $\llbracket l_t,r_t\rrbracket$ as long as it survives.
It was therefore sufficient to consider the two 0--1 edges at the boundary of the occupied set.
We observed that both edges show a positive bias toward the type-1 particle whenever $l_t \leq r_t$.
On $\mathbb{Z}^d$ with $d>1$, the configurations of type-1 particles can become more complex, which can give rise to 0--1 edges with no bias toward the type-1 particle, as we now show.

Assume that $\xi_t^0 \neq \varnothing$ and consider neighboring sites $v \in \xi_t^0$ (type-1) and $w \in \mathbb{Z}^d \setminus \xi_t^0$ (type-0). 
If $v$ has $k$ type-0 neighbors, $1 \leq k \leq 2d$, it dies and is replaced by an offspring of $w$ at rate $1/(k+(2d-k)\lambda)$.
If $w$ has $m$ type-1 neighbors, $1 \leq m \leq 2d$, it dies and is replaced by an offspring of $v$ at rate $\lambda/((2d-m)+m\lambda)$.
Irrespective of the values of $k$ and $m$, the edge between $v$ and $w$ always has a nonnegative bias toward $v$, since
\[
\lambda/((2d-m)+m\lambda)-1/(k+(2d-k)\lambda) \geq 1/2d-1/2d = 0.
\]
The bias is minimal when $k=2d$ and $m=2d$, in which case it is 0 (Fig \ref{fig:contribution}a).
For any configuration in which $k+m=2d$, the bias is
\begin{align} \label{eq:balanceddrift}
 \frac{\lambda/(k+m\lambda)-1/(k+m\lambda)}{\lambda/(k+m\lambda)+1/(k+m\lambda)} = \frac{\lambda-1}{1+\lambda} = \Delta,
\end{align}
the uniform drift of the random walk $(R_n^A)_{n \geq 0}$ embedded in the biased voter model (Section \ref{sec:bvmsurvprob}).
The bias toward $v$ is maximal when $k=1$ and $m=1$ (Fig \ref{fig:contribution}b), in which case it is
\begin{align}
\frac{\lambda/({(2d-1)+\lambda})-1/({1+(2d-1)\lambda})}{\lambda/({(2d-1)+\lambda})+1/({1+(2d-1)\lambda})} = \frac{(2d-1)(\lambda^2-1)}{2\lambda+(2d-1)(1+\lambda^2)} \label{eq:maximalcontr} 
> 
\Delta, \quad d>1.
\end{align}
Note that in one dimension, the zero-bias configuration never arises since $\xi_t^0$ is always a contiguous set of sites, and the maximal bias \eqref{eq:maximalcontr} is $\Delta$ for $d=1$.
As a ratio of $\Delta$, \eqref{eq:maximalcontr} becomes
\[
\frac{(2d-1)(1+\lambda)^2}{2\lambda+(2d-1)(1+\lambda^2)},
\]
which is decreasing in $\lambda$ for $\lambda >1$.
For fixed $d$, this ratio converges to $(2d-1)/d$ as $\lambda \to 1$ ($\beta \to 0$; weak selection), which in turn converges to $2$ as $d \to \infty$.
Thus, whereas in one dimension, the bias along a 0--1 edge is always positive and upper bounded by $\Delta$, higher dimensions give rise to 0--1 edges with biases in the range $[0,2\Delta)$.

\begin{figure}[t]
    \centering
    \includegraphics[scale=0.5]{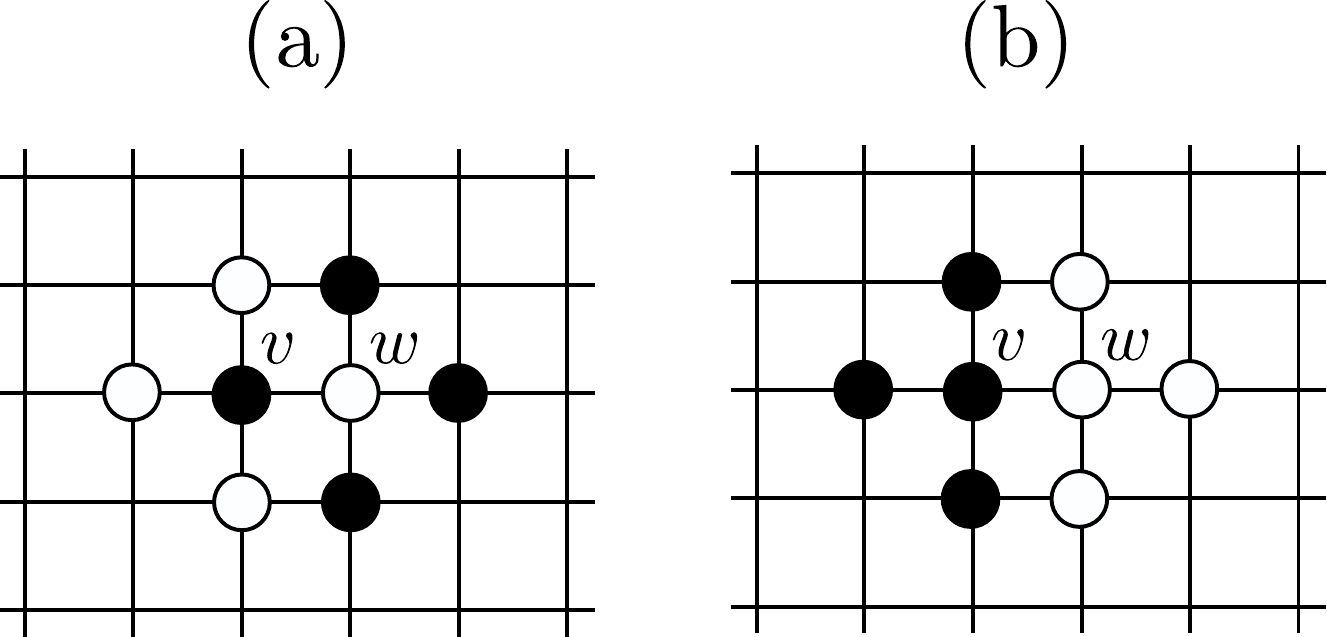}
    \caption{
    Configurations where a 0--1 edge shows minimal and maximal bias toward the type-1 particle, respectively.
    (a) Configuration where the edge between $v$ (type-1 particle) and $w$ (type-0 particle) does not favor the type-1 particle: All neighbors of $v$ are type-0 and all neighbors of $w$ are type-1.
    (b) Configuration where the edge between $v$ and $w$ maximally favors $v$: All neighbors of $v$ except $w$ are type-1, and all neighbors of $w$ except $v$ are type-0.}
    \label{fig:contribution}
\end{figure}

\subsection{Upper bound via embedded process} \label{sec:survprobupperbound}

Consider the jump process $(S_n^0)_{n \geq 0}$ embedded in $(|\xi_t^0|)_{t \geq 0}$.
When the death-birth process $(\xi_t^0)_{t \geq 0}$ is started with the origin occupied, the type-1 particle dies and becomes type-0 at rate 1, and each of the $2d$ type-0 neighbors dies and becomes type-1 at rate $\lambda/((2d-1)+\lambda)$.
The probability that the process survives the first type switch is therefore
\begin{align} \label{eq:p1}
p_1 = (2d\lambda)/\big((2d-1)+(2d+1)\lambda\big).
\end{align}
This is the transition probability $\P(S_{n+1}^0=2|S_n^0=1)$ for the embedded process, which applies any time there is a single type-1 particle in the system.
When $S_n^0 \geq 2$, the biases of 0--1 edges depend on the particular configuration of type-1 particles in the system, which implies that $(S_n^0)_{n \geq 0}$ is not a Markov chain on $\mathbb{Z}^d$ with $d>1$.
We can however obtain a simple upper bound on the survival probability of $(S_n^0)_{n \geq 0}$ by assuming that each 0--1 edge maximally favors the type-1 particle according to \eqref{eq:maximalcontr}.
This yields the transition probabilities
\begin{align} \label{eq:p_iub}
\overline{p}_i = \frac{\lambda\big(1+(2d-1)\lambda\big)}{2\lambda+(2d-1)(1+\lambda^2)}, \quad i \geq 2.    
\end{align}
We thus obtain the following upper bound, with $q_1 = 1-p_1$ and $\overline{q}_k = 1-\overline{p}_k$ for $k \geq 2$.

\begin{proposition} \label{prop:upperbound}
For any $d \geq 2$,
\begin{align}  \label{eq:dbupperbound}
    \P(\tau_\varnothing^0=\infty) &\leq 
    \textstyle \big(1+\frac{q_1}{p_1} \cdot \frac{\overline{p}_2}{\overline{p}_2-\overline{q}_2}\big)^{-1}
    = \big(1+ \frac{((2d-1)+\lambda)(1+(2d-1)\lambda)}{2d(2d-1)(\lambda^2-1)}\big)^{-1} =: \overline{\pi}_1.
\end{align}
\end{proposition}

It is straightforward to verify that
\begin{align*}
    & \overline{\pi}_1 \sim (2-1/d) \cdot \beta, \quad \beta \to 0, \\
        & \overline{\pi}_1 \to (2d)/(2d+1), \quad \beta \to \infty.
\end{align*}
In the weak-selection limit, the survival probability $\P(\tau_\varnothing^0=\infty)$ is of order at most $(2-1/d) \cdot \beta$.
Recall that in one dimension, the survival probability is of exact order $\beta$ as $\beta \to 0$ (Section \ref{sec:onedsurvprob}), and the same applies to the biased voter model in all dimensions (Section \ref{sec:bvmsurvprob}).
In the strong-selection limit, we obtain the upper bound $(2d)/(2d+1)$.
As we argued for the one-dimensional case (Section \ref{sec:onedsurvprob}),
this upper bound can also be obtained by noting that when the system is started by a single type-1 particle, this particle dies as frequently as its type-0 neighbors, which implies a lower bound of $1/(2d+1)$ on the extinction probability, independently of $\beta$.

\subsection{Lower bound via embedded process} \label{sec:survproblowerbound}

Obtaining a nontrivial lower bound for $\P(\tau_\varnothing^0=\infty)$ using $(S_n^0)_{n \geq 0}$ is not as straightforward, since as we noted in Section \ref{sec:biashighdim}, 0--1 edges can have no bias toward the type-1 particle.
We will show that if we only consider 0--1 edges involving type-0 particles in the unbounded component of $\mathbb{Z}^d \setminus \xi_t^0$, we obtain a minimal positive bias which guarantees that $(S_n^0)_{n \geq 0}$ escapes absorption with positive probability.
First, we need to introduce some notation.

\begin{notation*}
For $A,B \in \overline{\cal S}$
and $x, y \in \bZ^d$, let $d(x,y) := \|x-y\|_1$ denote the Manhattan distance between $x$ and $y$.
Set $d(x, A) := \inf_{y\in A} d(x,y)$ and $d(A,B) := \inf_{x\in A} d(x,B)$.  Let $\bar\partial A$ denote the exterior vertex boundary of $A$, that is
$$
\bar\partial A = \{ x\in \bZ^d\setminus A : d(x,A) = 1\},
$$
and let $\partial A$ denote the set of edges between $A$ and $\bar\partial A$. 
Let $\bar\partial^\infty A$ denote the set of vertices in $\bar\partial A$ that are in the unbounded component of $\bZ^d \setminus A$, and let $\partial^\infty A$ denote the set of edges between $A$ and $\bar\partial^\infty A$. 
\end{notation*}

If $\xi_t^0 \neq \varnothing$, each edge in $\partial^\infty \xi_t^0$ favors the type-1 particle, since the type-0 particle is not surrounded by type-1 particles.
The rate at which the type-1 particle replaces the type-0 particle exceeds the rate at which the opposite occurs by at least
\begin{align} \label{eq:C1}
C_1 = C_1(d,\lambda) := 
\frac\lambda{1+(2d-1)\lambda}-\frac1{2d} = \frac{\lambda-1}{2d(1+(2d-1)\lambda)}.
\end{align}
For a given configuration $\xi_t^0 = A$ of type-1 particles, a lower bound on the number of edges in $\partial^\infty A$ is obtained by the following isoperimetric inequality.

\begin{lemma}\label{isoperimetric}
Let $d\geq 2$.  Then for any finite set $A\subseteq \bZ^d$,
$$
\abs{\partial^\infty A} \geq \abs{\bar\partial^\infty A} \geq 2|A|^{1-1/d}.
$$
\end{lemma}

\begin{proof}
Section \ref{app:isoperimetric}.
\end{proof}

Assume that $|\xi_t^0|=k$.
Since each 0--1 edge has a nonnegative bias, it follows from \eqref{eq:C1} and Lemma \ref{isoperimetric} that the rate at which $|\xi_t^0| \to |\xi_t^0|+1$ exceeds the rate at which $|\xi_t^0| \to |\xi_t^0|-1$ by at least $2C_1 k^{1-1/d}$.
In addition, the rate at which $|\xi_t^0|$ jumps is trivially upper bounded by $(2d+1)k$.
Therefore,
the embedded process $(S_n^0)_{n \geq 0}$ is stochastically lower bounded by the simple random walk $(X_n)_{n \geq 0}$ on the integers with $X_0=1$ and absorption at zero, with transition probabilities $p_1$ given by \eqref{eq:p1} and
\begin{align*}
    \underline{p}_k = \textstyle \frac12 \big(1 + (2C_1/(2d+1)) k^{-1/d}\big), \quad k \geq 2.
\end{align*}
Set $C_2 := 2C_1/(2d+1)$. By exploiting the positive drift of $(X_n)_{n \geq 0}$, we can establish the following lower bound on the survival probability of $(\xi_t^0)_{t \geq 0}$.

\begin{proposition} \label{prop:lowerbound}
Define $I = I(d) := \int_0^\infty \exp\big({-x^{(d-1)/d}}\big)dx$.
Then, for any $d \geq 2$, 
\begin{align} \label{eq:dblowerbound}
    \textstyle \P(\tau_\varnothing^0 = \infty) \geq \big(
    (1/p_1) + (q_1/p_1) \cdot 
    I \cdot C_2^{-d/(d-1)} 
    \big)^{-1} =: \underline{\pi}_1.
\end{align}
\end{proposition}

\begin{proof}
Section \ref{app:lowerbound}.
\end{proof}

Define
\begin{align*}
    & K_0 = K_0(d) := I^{-1} \cdot \big(2d^2(2d+1)\big)^{-d/(d-1)}, \\
    &\textstyle  K_\infty = K_\infty(d) := \big(\frac{2d+1}{2d}+\frac{I}{2d} \cdot (d(4d^2-1))^{d/(d-1)}\big)^{-1}.
\end{align*}
It is straightforward to verify that
\begin{align*}
    & \underline{\pi}_1 \sim K_0 \cdot \beta^{d/(d-1)}, \quad \beta \to 0, \\
    &\underline{\pi}_1 \to \textstyle K_\infty, \quad \beta \to \infty.
\end{align*}
For the strong selection case, the constant is $K_\infty \approx 0.0022$ for $d=2$ and $K_\infty \approx 0.0042$ for $d=3$.
For the weak selection case, we get a lower bound of order $\beta^{d/(d-1)}$, which is of order $\beta^2$ for $d=2$ and $\beta^{3/2}$ for $d=3$, compared to the order $\beta$ upper bound established in Proposition \ref{prop:upperbound}.
The lower and upper bounds become of similar order when $d$ is large.

We conclude by remarking that the lower bound in \eqref{eq:dblowerbound} confirms that the event $\{\tau_\varnothing^0 = \infty\}$ has a positive probability for any $\beta >0$.
We state this as a corollary below.

\begin{corollary} \label{cor:survival}
For any $\beta >0$ and $d\geq 2$, $\P(\tau_\varnothing^0=\infty)>0$.
\end{corollary}

\section{Shape theorem for death-birth model on $\mathbb{Z}^d$, $d>1$} \label{sec:shapethm}

From Corollary \ref{cor:survival}, we know that the death-birth process $(\xi_t^0)_{t \geq 0}$ has a positive survival probability for any $\beta>0$.
The next question is how the process behaves when conditioned on survival.
In this section, we show that the Bramson-Griffeath shape theorem \eqref{eq:shapethmbvm} extends to the death-birth model.
Most properties of the biased voter model used in the proof of \eqref{eq:shapethmbvm} hold for the death-birth model, such as additivity, translation invariance and the strong Markov property.
However, the fact that the dual processes are different, and that the embedded process $(S_n^0)_{n \geq 0}$ is more complex, requires us to make a few adjustments to the proof.

\subsection{The process eventually contains a linearly-expanding ball}

In \cite{BramsonGriffeath81}, Bramson and Griffeath show that when the biased voter model is conditioned on nonextinction, it eventually contains a linearly-expanding ball with probability 1.
In their proof of Proposition 2 of \cite{BramsonGriffeath81}, they make explicit use of the graphical construction for the dual process $(\eta_t)_{t \geq 0}$ of the biased voter model.
Since the dual process $(\zeta_t)_{t \geq 0}$ for the death-birth model has a different structure, we need to modify this part of the proof.

On page 181 of \cite{BramsonGriffeath81}, for a given $x \in \Z^d$, Bramson and Griffeath define a Markov chain $(X_t^x)_{t \geq 0}$ embedded in the dual process of the biased voter model, which has a uniformly positive drift toward the origin whenever $\|X_t^x\| \geq \gamma$ for sufficiently large $\gamma$.
We wish to define an analogous Markov chain embedded in the dual process of the death-birth model.
To that end, suppose that $X_t^x = y \in \mathbb{Z}^d$.
Let $I(y) \subseteq {\cal N}(y)$ be the subset of neighbors $z$ of $y$ so that $\|z-y\|=1$ and {\em either} $\|z\|<\|y\|$ {\em or} $y=(y_1,\ldots,y_d)$ and $z=(z_1,\ldots,z_d)$ satisfy $y_i=0$, $z_i=1$ for some $i$.
Recall that in the graphical construction of $(\zeta_t^x)_{t \geq 0}$, for each subset $S \subseteq {\cal N}(y)$ of neighbors of $y$, $\delta$-arrows are drawn from all sites in $S$ to $y$ at rate $\nu_j$.
Using the graphical construction, we define the transitions that $X_t^x$ makes out of state $y$ as follows:
\begin{itemize}
    \item Suppose that $X_t^x = y$, and that for some $S \subseteq {\cal N}(y)$, $\delta$-arrows are drawn from the sites in $S$ to $y$. If $S \cap I(y) \neq \varnothing$, $X_t^x$ jumps from $y$ to a randomly chosen neighbor in $S \cap I(y)$. If $S \cap I(y) = \varnothing$, $X_t^x$ jumps from $y$ to a randomly chosen neighbor in $S$.
\end{itemize}
By construction, $X_t^x \in \zeta_t^x$ for all $t \geq 0$, so $(X_t^x)_{t \geq 0}$ is embedded in $(\zeta_t^x)_{t \geq 0}$.
We next need to determine the rates at which $X_t^x$ jumps to sites in $I(y)$ versus sites in ${\cal N}(y) \setminus I(y)$.
Note that for $X_t^x$ to jump to a neighbor in ${\cal N}(y) \setminus I(y)$, the set $S$ must only contain sites in ${\cal N}(y) \setminus I(y)$.
Since $|I(y)|=d$, $X_t^x$ jumps to a given neighbor $z \in {\cal N}(y) \setminus I(y)$ at rate
\begin{align*}
    \textstyle 
    \frac1d \sum_{k=1}^d \nu_k \binom{d}k = \frac1d \frac1{1+\lambda},
\end{align*}
where the equality follows from \eqref{eq:linearsystem}.
Since $X_t^x$ jumps out of state $y$ at total rate 1, it jumps to a given neighbor $z \in I(y)$ at rate $\frac1d \frac\lambda{1+\lambda}$.
If we run time at rate $(1+\lambda)/2$, $X_t^x$ jumps to a neighbor $z \in I(y)$ at rate $\lambda/(2d)$, and to a neighbor $z \in {\cal N}(y) \setminus I(y)$ at rate $1/(2d)$.
These are the jump rates for the Markov chain $(X_t^x)_{t \geq 0}$ defined in expression (23) on page 181 of \cite{BramsonGriffeath81}.

By redefining $(X_t^x)_{t \geq 0}$ as above and rescaling time, we can use Bramson and Griffeath's arguments to prove that the death-birth model eventually contains a linearly expanding ball.
Below, we state a refinement of the main result in \cite{BramsonGriffeath81}, which is stated as Proposition 1 in \cite{BraGri80}.
This refinement is needed to establish the existence of an asymptotic shape in Section \ref{sec:BraGri2} below.
Here, $D_r := \{x \in \Z^d: \|x\| \leq r\}$ denotes the Euclidean ball of radius $r$.

\begin{proposition} \label{prop:linexpball}
There exist positive constants $c,A,\alpha,z_0$ so that
\begin{align} \label{eq:linexpball}
\P\big( D_{ct} \subseteq \xi_{t+z^2}^0,\; t\geq 0 \big| \tau_\varnothing^0=\infty\big) \geq 1-Ae^{-\alpha z}, \quad z \geq z_0.
\end{align}
\end{proposition}

\subsection{Existence of an asymptotic shape} \label{sec:BraGri2}

In \cite{DurrettGriffeath82}, Durrett and Griffeath define growth models on $\mathbb{Z}^d$, and they give sufficient conditions for the existence of a shape theorem for such models.
For the death-birth model, two conditions need to be verified, one of which follows immediately from Proposition \ref{prop:linexpball}.
The other condition states that there exist positive constants $\gamma, c, C, p$ such that
\begin{align*}
    \P(t < \tau_\varnothing^0 < \infty) \leq C e^{-\gamma t^p}, \quad t \geq 0.
\end{align*}
This condition follows from Lemmas \ref{lemma:extinctionproblargesetdb} and \ref{lemma:cannothangaroundatsmallsize} below, as we show in Proposition \ref{prop:shapethmdb}.

In Lemma \ref{lemma:extinctionproblargesetdb}, we show that if the death-birth process is started by $k$ type-1 particles, its extinction probability decreases exponentially fast in $k^p$ for some $p>0$, independently of how the initial type-1 particles are configured.
This property is easy to establish for the biased voter model using the embedded random walk $(R_n^0)_{n \geq 0}$, as we demonstrated in \eqref{eq:extinctionprobabilitybvm} of Section \ref{sec:bvmsurvprob}.
The property also holds trivially for the death-birth model in one dimension, due to the simple geometry of the occupied set, see \eqref{eq:extinctionprobabilityoned} of Section \ref{sec:onedshape}.

\begin{lemma} \label{lemma:extinctionproblargesetdb}
There are constants $C,\gamma>0$ so that
\[
\textstyle \sup_{A \in \overline{\cal S}, \abs{A} = k} \P(\tau_\varnothing^A < \infty) \leq C \exp\big(\!-\!\gamma k^{(d-1)/d}\big), \quad k \geq 1.
\]
\end{lemma}

\begin{proof}
Section \ref{app:shapethm1}.
\end{proof}

In Lemma \ref{lemma:cannothangaroundatsmallsize}, we show that the probability that the death-birth process remains alive at a size smaller than order $t^{d/(d+1)}$ up until time $t$ decreases exponentially fast in $t^{(d-1)/(d+1)}$.
Again, a similar property is easy to establish for the biased voter model using the embedded random walk $(R_n^0)_{n \geq 0}$.
In that case, $t^{d/(d+1)}$ and $t^{(d-1)/(d+1)}$ can be replaced by $t$.

\begin{lemma} \label{lemma:cannothangaroundatsmallsize}
For sufficiently small $\varepsilon>0$, there are constants $C,\gamma>0$ so that
\[
\P\Big(|\xi_s^0| \in \big(0,\varepsilon t^{d/(d+1)}\big), \; s \leq t\Big) \leq C \exp\big(\!-\!\gamma t^{(d-1)/(d+1)}\big), \quad t \geq 0.
\]
\end{lemma}

\begin{proof}
Section \ref{app:shapethm2}.
\end{proof}

Finally, in Proposition \ref{prop:shapethmdb}, we state the shape theorem for the death-birth model.
In the proof, we give a brief outline of how it follows from Proposition \ref{prop:linexpball} and Lemmas \ref{lemma:extinctionproblargesetdb} and \ref{lemma:cannothangaroundatsmallsize}.

\begin{proposition} \label{prop:shapethmdb}
There exists a convex subset $D$ of $\R^d$ so that for each $\varepsilon>0$,
\begin{align*}
    \P\big(\exists t_*<\infty: (1-\varepsilon)tD \cap \mathbb{Z}^d \subseteq \xi_t^0 \subseteq (1+\varepsilon)tD,\; t\geq t_*  \big| \tau_\varnothing^0=\infty\big)=1.
\end{align*}
\end{proposition}

\begin{proof}
Section \ref{app:shapethm3}.
\end{proof}

\section{Alternative birth-death and death-birth models} \label{sec:altmodels}

So far, we have discussed the biased voter model, where birth precedes death and fitness affects the birth event (${\rm B}^{\rm f}{\rm D}$ model), and the death-birth model, where death precedes birth and fitness affects the birth event (${\rm DB}^{\rm f}$ model).
From a biological standpoint, the analogous models where fitness affects the death event are equally relevant.
For example, in the context of cancer initiation, cells can acquire resistance to apoptosis by loss of p53 function or by upregulation of anti-apoptotic Bcl2 \cite{kaveh2015duality}.
In this section, we briefly discuss these alternative models, including models where fitness affects both the birth event and the death event.

\subsection{${\bf D}^{\bf f}{\bf B}$ model} \label{sec:Dfbmodel}

We first consider a death-birth model where fitness affects the death event (${\rm D}^{\rm f}{\rm B}$ model).
In this model, type-0 particles die at rate 1, and type-1 particles die at rate $\kappa = 1/\lambda$ with $\lambda=1+\beta$ and $\beta>0$.
When a particle dies at $x \in \mathbb{Z}^d$, a neighboring particle is chosen uniformly at random to give birth and replace the particle at $x$.
We can construct this process graphically as follows:
For each $x \in \mathbb{Z}^d$ and each nearest neighbor $y$ of $x$, draw a $\delta$-arrow from $y$ to $x$ at rate $(1/2d)\kappa = 1/(2d\lambda)$, and draw a regular arrow at rate $(1/2d)(1-\kappa) = (\lambda-1)/(2d\lambda)$.
If we run this process at rate $1/\kappa=\lambda$, we obtain the graphical construction of the biased voter model (Section \ref{sec:bvmgraphical}).
The ${\rm D}^{\rm f}{\rm B}$ model is therefore a sped-up version of the ${\rm B}^{\rm f}{\rm D}$ model.

\subsection{${\bf D}^{\bf f}{\bf B}^{\bf f}$ model} \label{sec:DfBfmodel}

We next consider a generalization of the ${\rm D}^{\rm f}{\rm B}$ and ${\rm D}{\rm B}^{\rm f}$ models,
where we allow fitness to affect both the death event and the birth event.
In the ${\rm D}^{\rm f}{\rm B}^{\rm f}$ model, there are two fitness parameters $\beta_1>0$ and $\beta_2>0$.
Type-0 particles die at rate 1 and type-1 particles die at rate $1/\lambda_1 := 1/(1+\beta_1)$.
On the subsequent birth event, type-1 particles have proliferation fitness $\lambda_2 := 1+\beta_2$,
while type-0 particles have fitness $1$.
In other words, upon a death at $x \in \Z^d$, a neighboring particle is selected to give birth and replace the particle at $x$ with probability proportional to its proliferation fitness.

To construct a graphical representation of the ${\rm D}^{\rm f}{\rm B}^{\rm f}$ model, we combine the constructions for the ${\rm D}^{\rm f}{\rm B}$ and ${\rm D}{\rm B}^{\rm f}$ models (Sections \ref{sec:Dfbmodel} and \ref{sec:gencasegraphicalsubsec}).
First, define
\begin{align} \label{eq:linearsystemsolutiondfbf}
\widehat{\nu}_j := \frac{j! \beta_2^{j-1}}{\prod_{k=1}^j (k+(2d-k)\beta_2)}
\end{align}
following \eqref{eq:linearsystemsolution}.
Then, for each $x \in \mathbb{Z}^d$ and each $S \subseteq {\cal N}(x)$ with $|S|=j$, draw $\delta$-arrows from all sites in $S$ to $x$ simultaneously at rate $\frac1{1+\beta_1} \widehat{\nu}_j$, and draw regular arrows at rate $\frac{\beta_1}{1+\beta_1} \widehat{\nu}_j$.
In the dual process, for a given particle at $x \in \mathbb{Z}^d$ and each $S \subseteq {\cal N}(x)$ with $|S|=j$, at rate $\frac1{1+\beta_1} \widehat{\nu}_j$, the particle at $x$ splits into $j$ particles and places one offspring at each site $y \in S$, and at rate $\frac{\beta_1}{1+\beta_1} \widehat{\nu}_j$, the particle remains at $x$ and places $j$ copies of itself at each $y \in S$.

In one dimension, the survival probability for the ${\rm D}^{\rm f}{\rm B}^{\rm f}$ model, starting with the origin occupied, is given by
\begin{align*}
    \frac{2(\lambda_1\lambda_2-1)}{2\lambda_1\lambda_2+\lambda_2-1},
\end{align*}
which reduces to \eqref{eq:survprobsol} for $\lambda_1=1$.
In higher dimensions, every 0--1 edge has a positive bias toward the type-1 particle,
in contrast to the ${\rm DB}^{\rm f}$ model.
For the minimal-bias configuration in Figure \ref{fig:contribution}a,
the type-0 particle at $w$ dies and is replaced by an offspring of $v$ at rate $1/(2d)$, while the type-1 particle at $v$ dies and is replaced by an offspring of $w$ at rate $(1/\lambda_1)(1/(2d))$.
The bias toward the type-1 particle is therefore
\begin{align*}
    \frac{1/(2d)-(1/\lambda_1)(1/(2d))}{1/(2d) + (1/\lambda_1)(1/(2d))} = \frac{\lambda_1-1}{1+\lambda_1} > 0.
\end{align*}
This is the bias of 0--1 edges for $\beta_2=0$, in which case the ${\rm D}^{\rm f}{\rm B}^{\rm f}$ model reduces to the biased voter model.
It follows that the survival probability for the ${\rm D}^{\rm f}{\rm B}^{\rm f}$ model is lower bounded by $\beta_1/(1+\beta_1)$, see \eqref{eq:survprobbvm}.
In particular, it is positive for any $\beta_1 >0$.
We can otherwise obtain lower and upper bounds analogous to Propositions \ref{prop:upperbound} and \ref{prop:lowerbound} by taking 
\begin{align*}
    & p_1 = \frac{2d\lambda_1\lambda_2}{(2d-1)+\lambda_2+2d\lambda_1\lambda_2}, \\
    & \overline{p}_k = \frac{\lambda_1\lambda_2\big(1+(2d-1)\lambda_2\big)}{(2d-1)+\lambda_2+\lambda_1\lambda_2+(2d-1)\lambda_1\lambda_2^2}, \quad k \geq 2, \\
    & C_1 = \frac{\lambda_2}{1+(2d-1)\lambda_2} - \frac1{2d\lambda_1}.
\end{align*}
For $\lambda_1=1$, these quantities reduce to the corresponding quantities in Section \ref{sec:DBsurvprob}.

Denote the ${\rm D}^{\rm f}{\rm B}^{\rm f}$ process started from $A \in \overline{\cal S}$ by $(\varphi_t^A)_{t \geq 0}$.
Note that by eliminating the regular arrows from the graphical construction and running time at rate $1+\beta_1$,
we can couple $(\varphi_t^A)_{t \geq 0}$ with a ${\rm D}{\rm B}^{\rm f}$ model $(\xi_t^A)_{t \geq 0}$ so that $\xi_t^A \subseteq \varphi_{(1+\beta_1)t}^A$ for all $t \geq 0$.
Under this coupling, Proposition \ref{prop:linexpball} and Lemma \ref{lemma:extinctionproblargesetdb}  extend easily to the ${\rm D}^{\rm f}{\rm B}^{\rm f}$ model.
For Lemma \ref{lemma:cannothangaroundatsmallsize}, it suffices to note that the discrete-time jump process embedded in $(|\varphi_t^A|)_{t \geq 0}$ is stochastically lower bounded by the embedded random walk $(R_n^A)_{n \geq 0}$ for the biased voter model with $\beta = \beta_1$ (Section \ref{sec:bvmsurvprob}).
In conclusion, our analysis of the survival probability and the asymptotic shape for the ${\rm D}{\rm B}^{\rm f}$ model extends relatively easily to the more general ${\rm D}^{\rm f}{\rm B}^{\rm f}$ model.

\subsection{${\bf B}{\bf D}^{\bf f}$ model} \label{sec:BDfmodel}

We now consider a birth-death model where fitness affects the death event (${\rm B}{\rm D}^{\rm f}$ model).
In this model, type-0 particles have fitness 1 and type-1 particles have fitness $\lambda=1+\beta$ with $\beta>0$.
A particle at $x \in \mathbb{Z}^d$ gives birth to a new particle at rate 1, and a neighbor is selected to die with probability  inversely proportional to its fitness.
In other words, if $x$ gives birth when it has $i$ type-0 neighbors and $j$ type-1 neighbors, a type-0 neighbor is selected to die with probability $i/(i+j(1/\lambda)) = i\lambda/(i\lambda+j)$, and a type-1 neighbor is selected to die with probability $j(1/\lambda)/(i+j(1/\lambda)) = j/(i\lambda+j)$.

\begin{figure}[t]
    \centering
    \includegraphics[scale=0.45]{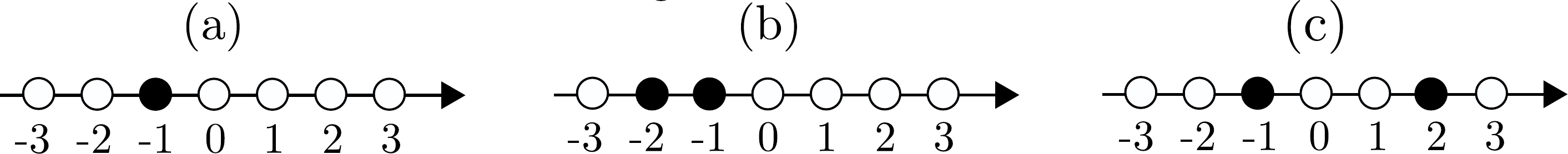}
    \caption{
    In the ${\rm BD}^{\rm f}$ model on $\mathbb{Z}$, the rate at which a type-0 particle (white) at the origin switches to type-1 (black) depends both on its neighbors and the neighbors of its neighbors.
    {\bf (a)} If the type-1 particle at $-1$ has a type-0 neighbor to its left, it replaces the particle at $0$ with rate $1/2$.
    {\bf (b)} If the type-1 particle at $-1$ has a type-1 neighbor to its left, the rate becomes $1/(1+\lambda)>1/2$.
    {\bf (c)} To determine the rate at which the type-0 particle at $0$ switches to type-1, it is not sufficient to know the number of neighbors and neighbors of neighbors of each type.
    For example, if there is a type-0 particle at $-2$ and type-1 particle at $2$, the type-1 particle at $-1$ replaces the particle at $0$ with rate $1/2$, which is different from the rate in ${\bf (b)}$.    }
    \label{fig:nographicalrepr}
\end{figure}

In the models we have investigated so far, the rate at which a site $x \in \mathbb{Z}^d$ switches from type-0 to type-1, or vice versa, has been determined by the number of neighbors of each type, see Sections \ref{sec:bvmdescr} and \ref{sec:dbmodeldef}.
The ${\rm BD}^{\rm f}$ model is fundamentally different in that the switching rate depends on the neighbors of neighbors of $x$.
For example, consider the configurations in Figure \ref{fig:nographicalrepr} for the ${\rm BD}^{\rm f}$ model on $\Z$, where white particles are type-0 and black particles are type-1.
In Figure \ref{fig:nographicalrepr}a, the type-1 particle at $-1$ gives birth and replaces the type-0 particle at the origin at rate $1/2$, but in Figure \ref{fig:nographicalrepr}b, this occurs at rate $1/(1+1/\lambda)>1/2$.
Moreover, to determine the switching rate for a site $x$, it is not sufficient to know the number of neighbors and neighbors of neighbors of $x$ of each type.
For example, in Figures \ref{fig:nographicalrepr}b and \ref{fig:nographicalrepr}c, the type-0 particle at $0$ has one type-1 neighbor and one type-1 neighbor of neighbor, but the rate at which the type-0 particle switches to type-1 is different for each case.

By a similar reasoning, we see that the ${\rm BD}^{\rm f}$ model is not an additive process, in contrast to the processes considered above.
Denote the ${\rm B}{\rm D}^{\rm f}$ model on $\mathbb{Z}$ by $(\psi_t^A)_{t \geq 0}$, starting from the set $A \in \overline{\cal S}$.
At time 0 in $(\psi_t^{\{-1\}})_{t \geq 0}$, the origin becomes occupied at rate $1/2$ (Fig \ref{fig:nographicalrepr}a), and the same is true of the process $(\psi_t^{\{-2\}} \cup \psi_t^{\{-1\}})_{t \geq 0}$.
In $(\psi_t^{\{-2,-1\}})_{t \geq 0}$, however, the origin becomes occupied at rate $1/(1+1/\lambda)>1/2$ (Fig \ref{fig:nographicalrepr}b).
This simple example shows that $(\psi_t^{\{-2\}})_{t \geq 0}$, $(\psi_t^{\{-1\}})_{t \geq 0}$ and $(\psi_t^{\{-2,-1\}})_{t \geq 0}$ cannot be simultaneously coupled so that
\[
\psi_t^{\{-2,-1\}} = \psi_t^{\{-2\}} \cup \psi_t^{\{-1\}}, \quad t \geq 0,
\]
meaning that $(\psi_t^A)_{t \geq 0}$ is not additive.
It follows that the ${\rm B}{\rm D}^{\rm f}$ model cannot be represented graphically using the basic tools introduced above, that is, by drawing arrows and $\delta$-arrows \cite{durrett1988lecture,durrett95,gray1986duality}.

In one dimension, the survival probability for the ${\rm B}{\rm D}^{\rm f}$ model, starting with the origin occupied, can be easily computed as
\begin{align} \label{eq:BDfmodelsurvproboned}
    \frac{(\lambda+1)(\lambda-1)}{\lambda^2+2\lambda-1}.
\end{align}
In Figure \ref{fig:ratios_comparison}, we show \eqref{eq:BDfmodelsurvproboned} as a proportion of the survival probability \eqref{eq:survprobbvm} for the biased voter model, and we compare it with the proportion between the survival probability \eqref{eq:survprobsol} for the ${\rm DB}^{\rm f}$ model and for the biased voter model.
Note first
that the survival probability for the biased voter model dominates that of the ${\rm BD}^{\rm f}$ and ${\rm DB}^{\rm f}$ models.
Then note that while the proportion for the ${\rm DB}^{\rm f}$ model decreases from 1 to 2/3 as $\beta \to \infty$, it is nonmonotonic for the ${\rm BD}^{\rm f}$ model, and it converges to 1 both as $\beta \to 0$ and $\beta \to \infty$.
Thus, the survival probability for the ${\rm BD}^{\rm f}$ model is the same as for the biased voter model in the limits of weak and strong selection.
The proportion for the ${\rm BD}^{\rm f}$ model is smallest $(1/4)(2+\sqrt{2}) \approx 0.85$ at $\beta=\sqrt{2}$.

For the survival probability in higher dimensions, we note that our analysis in Section \ref{sec:DBsurvprob} does not rely on the model having a graphical representation.
The configurations of type-0 and type-1 particles giving rise to the minimal and maximal biases of 0--1 edges are the same (Fig \ref{fig:contribution}), and we can obtain the bounds in Propositions \ref{prop:upperbound} and \ref{prop:lowerbound} with the substitutions
\begin{align*}
    & p_1 = \frac{1+(2d-1)\lambda}{(2d+1)+(2d-1)\lambda}, \\
    & \overline{p}_k = \frac{\lambda\big(1+(2d-1)\lambda\big)}{(2d-1)+2\lambda+(2d-1)\lambda^2}, \quad k \geq 2, \\
    & C_1 = \frac1{2d}-\frac1{(2d-1)+\lambda}.
\end{align*}

In one dimension, we can easily obtain the shape theorem \eqref{eq:shapethemoned}.
However, we are not able to obtain a shape theorem in higher dimensions using the same argument as in Section \ref{sec:shapethm}, since the argument relies on the additivity of the process and the associated dual process.

\begin{figure}[t]
    \centering
    \includegraphics[scale=1]{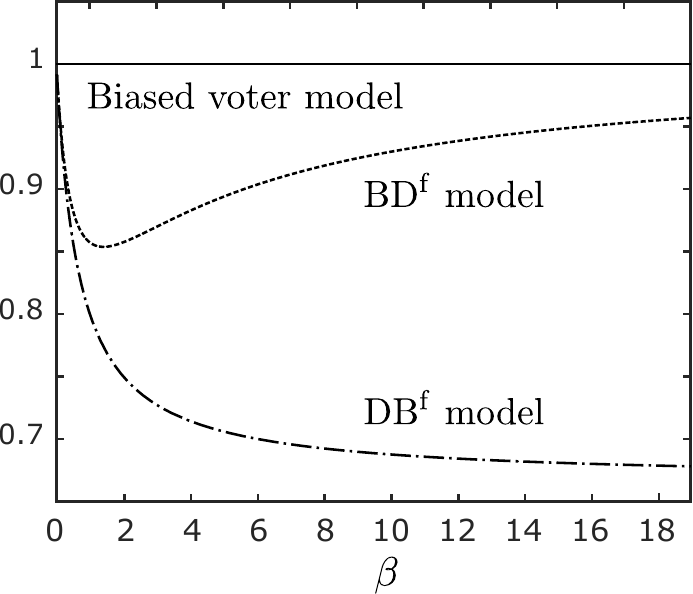}
    \caption{
    Comparison of the survival probabilities for the biased voter model (${\rm B}^{\rm f}{\rm D}$ model), death-birth model (${\rm DB}^{\rm f}$ model) and ${\rm BD}^{\rm f}$ model in one dimension (on $\mathbb{Z}$), starting with the origin occupied.
    The probabilities are shown as a proportion of the survival probability for the biased voter model.
    }
    \label{fig:ratios_comparison}
\end{figure}

Finally, we can define a ${\rm B}^{\rm f}{\rm D}^{\rm f}$ model, where type-0 particles give birth at rate 1 and type-1 particles give birth at rate $\lambda_1=1+\beta_1$ with $\beta_1>0$.
Then, on the subsequent death event, neighboring type-0 particles have fitness 1 and neighboring type-1 particles have fitness $\lambda_2=1+\beta_2$ with $\beta_2>0$.
To avoid repetition, we will not discuss this model further, other than to remark that it can be analyzed in a similar way to the ${\rm BD}^{\rm f}$ model.

\subsection{Conclusions}

The preceding discussion reveals several interesting differences between the models considered, depending on the order of birth and death events and how fitness affects the dynamics:
\begin{itemize}
    \item When fitness affects the former updating event, the death-birth (${\rm D}^{\rm f}{\rm B}$) and birth-death (${\rm B}^{\rm f}{\rm D}$) models are equivalent.
    \item When fitness affects the latter updating event, the death-birth (${\rm DB}^{\rm f}$) and birth-death $({\rm BD}^{\rm f})$ models are not only not equivalent, the birth-death model is not additive. In this respect, the ${\rm BD}^{\rm f}$ model is fundamentally different from the ${\rm B}^{\rm f}{\rm D}$, ${\rm D}^{\rm f}{\rm B}$ and ${\rm DB}^{\rm f}$ models.
    \item Since the ${\rm D}^{\rm f}{\rm B}$ model is equivalent to the biased voter model, our analysis for the ${\rm DB}^{\rm f}$ model extends relatively easily to the more general ${\rm D}^{\rm f}{\rm B}^{\rm f}$ model.
    \item 0--1 edges with no bias toward the type-1 particle arise only in the ${\rm DB}^{\rm f}$ and ${\rm BD}^{\rm f}$ models, where fitness affects the latter event only.
\end{itemize}

\section{Open problems: Death-birth model in weak-selection limit} \label{sec:openproblems}

In one dimension, the death--birth model $({\rm DB}^{\rm f}$ model) has survival probability $(2\beta)/(2+3\beta)$ when started from a single type-1 particle (Section \ref{sec:onedsurvprob}).
In the limit of weak selection ($\beta \to 0$), the probability is of order $\beta$.
In higher dimensions, the survival probability is bounded by $c \beta^{d/(d-1)}$ and $C\beta$ as $\beta \to 0$ for some $c,C>0$, as we showed in Sections \ref{sec:survprobupperbound} and \ref{sec:survproblowerbound}.
For the biased voter model, the survival probability is 
$\beta/(1+\beta)$, which is 
of order $\beta$ as $\beta \to 0$ (Section \ref{sec:bvmsurvprob}).
While the survival probabilities for the death-birth and biased voter model are in general distinct, an interesting open problem is to determine whether they are the same in the weak--selection limit.
We know that this is the case in one dimension, and in higher dimensions, the lower bound $c \beta^{d/(d-1)}$ for the death-birth model approaches the order of the upper bound $C\beta$ as $d \to \infty$.

Recall that in one dimension, the embedded jump process $(S_n^0)_{n \geq 0}$ has a positive drift of $(\lambda-1)/(1+3\lambda)$ when $S_n^0=1$ and $\Delta = (\lambda-1)/(1+\lambda)$ when $S_n^0 \geq 2$, where $\Delta$ is the uniform drift of the analogous jump process for the biased voter model (Section \ref{sec:onedsurvprob}).
The difference between the survival probabilities for the two processes on $\mathbb{Z}$ can therefore be traced to different behaviors when there is a single type-1 particle in the system.
Motivated by this fact, Kaveh et al.~\cite{kaveh2015duality} have suggested an approximation of the survival probability in two dimensions ($d=2$) by taking $p_1$ as in \eqref{eq:p1} and setting $p_k = \lambda/(1+\lambda)$ for $k \geq 2$. This leads to the approximation
\begin{align*}
    \frac{2d(\lambda-1)}{(2d+1)\lambda-1} = \frac{2d\beta}{2d+(2d+1)\beta}.
\end{align*}
In their Figure 7, Kaveh et al.~present simulation results for $d=2$ and several values of $\beta \in [0,3]$, which indicate that this approximation is reasonable.
We cannot conclude that the approximation has the correct asymptotics in the $\beta \to 0$ limit, but it is consistent with the hypothesis that the survival probability is of order $\beta$ as $\beta \to 0$.

While the dual processes for the death-birth model and the biased voter model are seemingly very different, they turn out to behave similarly in the weak-selection limit.
In the dual process $(\zeta_t^B)_{t \geq 0}$ for the death-birth model, each particle splits into $j$ particles at rate $\alpha_j := \binom{2d}j \nu_j$, $1 \leq j \leq 2d$, and the $j$ particles are placed at a randomly chosen subset of neighbors of size $j$.
It is straightforward to verify that for small $\beta$,
\begin{align} \label{eq:ratesdualweakselection}
\begin{split}
    & \textstyle \alpha_1 = 1 - \frac{2d-1}{2d} \beta + O(\beta^2), \\
    & \textstyle \alpha_2 = \frac{2d-1}{2d} \beta + O(\beta^2), \\
    & \textstyle \alpha_j = O(\beta^2), \quad 3 \leq  j \leq 2d.    
\end{split}
\end{align}
Thus, as $\beta \to 0$, particles in the dual process jump at rate $\sim 1$, they split into two particles at rate $\sim \frac{2d-1}{2d} \beta$, and they split into three or more particles much less frequently.
When a particle splits into two particles, the two new particles are placed at neighboring sites of the parent particle, meaning that they start two lattice locations apart.

In the dual $(\eta_t^B)_{t \geq 0}$ for the biased voter model, particles jump at rate 1, and they give birth to a new particle at rate $\beta$, which is placed at a neighboring site. 
When a particle gives birth, the parent-daughter pair coalesces on the first jump of parent or daughter with probability $\frac1{2d}$.
Conversely, the pair avoids coalescence on the first jump with probability $\frac{2d-1}{2d}$, in which case the parent and daughter become two lattice locations apart.
Thus, if we ignore any new particle which coalesces with its parent on the first jump of either particle, the dual process $(\eta_t^B)_{t \geq 0}$  produces new particles at rate $\frac{2d-1}{2d} \beta$, and a parent-daughter pair will be two lattice locations apart after the first transition of either particle.
Thus, the particle production dynamics of the dual processes $(\zeta_t^B)_{t \geq 0}$ and $(\eta_t^B)_{t \geq 0}$ become similar as $\beta \to 0$.
This in turn indicates that the survival probabilities for the death-birth and biased voter model are similar, since using the duality relation \eqref{eq:dualitydeathbirth}, the survival probability of the death-birth model can be related to an occupation probability for the dual process via
\begin{align*}
    \textstyle \P(\tau_\varnothing^0 = \infty) = \lim_{t \to \infty} \P(\xi_t^0 \neq \varnothing) = \lim_{t \to \infty} \P(0 \in \zeta_t^{\mathbb{Z}}). 
\end{align*}

In \cite{durrett2016spatial}, the propagation speed of the biased voter model on $\mathbb{Z}^d$, that is, the radius of the asymptotic shape $\widetilde D$ in \eqref{eq:shapethmbvm}, is calculated in the weak-selection limit using the duality relation \eqref{eq:dualityrelationbvm}.
Another interesting open problem is to determine whether the similarity between the dual processes $(\zeta_t^B)_{t \geq 0}$ and $(\eta_t^B)_{t \geq 0}$ as $\beta \to 0$ induces the same propagation speed for the death-birth model as for the biased voter model.

\section{Proofs} \label{sec:proofs}

\subsection{Solution to the linear system (\ref{eq:linearsystem})} \label{app:linearsystem}

Here, we verify by induction that \eqref{eq:linearsystemsolution} solves \eqref{eq:linearsystem}. 
Note first that $\nu_1$ as given by \eqref{eq:linearsystemsolution} satisfies the first equation in \eqref{eq:linearsystem}.
Assume that~\eqref{eq:linearsystemsolution} holds for $\nu_1, \ldots, \nu_j$ for some $j=1,\ldots,2d-1$.  Then,
\begin{align*}
\nu_{j+1} &= \textstyle \frac{j+1}{j+1+(2d-j-1)\lambda} - \sum_{\ell=1}^j {j+1 \choose \ell} \nu_\ell \\
&= \textstyle \frac{j+1}{j+1+(2d-j-1)\lambda} - \sum_{\ell=1}^j {j+1 \choose \ell} \cdot \frac{\ell!(\lambda-1)^{\ell-1}}{\prod_{k=1}^\ell (k+(2d-k)\lambda)} \\
&= \textstyle \frac{1}{D_1^{j+1}} \underbrace{\textstyle \Big((j+1) D_1^j - \sum_{\ell=1}^j (j+1)\cdots(j+2-\ell) \cdot D_{\ell+1}^{j+1} \cdot (\lambda-1)^{\ell-1} \Big)}_{(\ast)},
\end{align*}
where $D_i^j = \prod_{k=i}^j (k + (2d-k)\lambda)$.
We are finished once we verify that $(\ast)$ is equal to $(j+1)! (\lambda-1)^j$.
This follows by inductively combining the positive term inside the parenthesis with the first term in the sum on the right, and observing that for each $i = 1, \ldots, j$,
$$
D_i^j - D_{i+1}^{j+1} = (j+1-i) (\lambda-1) D_{i+1}^j,
$$
which implies
$$
(j+1)\cdots(j+2-i) \cdot (\lambda-1)^{i-1} \cdot \big(D_i^j - D_{i+1}^{j+1}\big) = (j+1)\cdots(j+1-i)\cdot (\lambda-1)^i \cdot D_{i+1}^j.
$$

\subsection{Proof of Lemma \ref{isoperimetric}} \label{app:isoperimetric}

\begin{proof}
For a finite set $A\subseteq \bZ^d$, let $A_i$ denote the projection of $A$ along the $i$-th coordinate axis, that is, $A_i = \{(x_1, \ldots, x_{i-1}, x_{i+1}, \ldots, x_{d}) : (x_1,\ldots, x_{d}) \in A\}$. Then, by the Loomis-Whitney inequality (Theorem 2 of \cite{loomis1949inequality}),
\[
\textstyle \prod_{i=1}^d \abs{A_i} \geq \abs{A}^{d-1},
\]
which implies that there exists $j \in \{1, \ldots, d\}$ such that $\abs{A_j} \geq \abs{A}^{(d-1)/d}$.
Now, each element $(x_1,\ldots,x_{j-1},x_{j+1},x_d)$ of $A_j$ corresponds to at least two elements of $\bar\partial^\infty A$.  Namely, if 
\begin{align*}
\alpha =& \min \{y : (x_1, \ldots, x_{j-1}, y, x_{j+1}, \ldots, x_{d}) \in A\}, \\
\beta =& \max \{y : (x_1, \ldots, x_{j-1}, y, x_{j+1}, \ldots, x_{d}) \in A\},
\end{align*}
then $(x_1, \ldots, x_{j-1}, \alpha-1, x_{j+1}, \ldots, x_{d}) \in \bar\partial^\infty A$ and $(x_1, \ldots, x_{j-1}, \beta+1, x_{j+1}, \ldots, x_{d}) \in \bar\partial^\infty A$.  This proves the lemma.
\end{proof}

\subsection{Proof of Proposition \ref{prop:lowerbound}} \label{app:lowerbound}

\begin{proof}
Let $(X_n)_{n \geq 0}$ denote the simple random walk on the integers with $X_0=1$ and absorption at zero, with jump probabilities
$p_1$ according to \eqref{eq:p1}, and 
\begin{equation*}
\textstyle \underline{p}_k = \frac{1}{2} (1 + C_2k^{-1/d}), \quad k \geq 2,
\end{equation*}
where $C_2 := 2C_1/(2d+1)>0$ with $C_1$ defined as in \eqref{eq:C1}.
Set ${q}_1 := 1-{p}_1$ and $\underline{q}_k := 1-\underline{p}_k$ for $k \geq 2$.
By a generalized version of the gambler's ruin formula (Theorem 5.3.11 of \cite{MR3930614}), the survival probability for $(X_n)_{n \geq 0}$ is
\begin{align}  \label{eq:survproblb}
    \textstyle \big(1+\sum_{\ell=1}^\infty \frac{q_1 \underline{q}_2 \cdots \underline{q}_\ell}{p_1 \underline{p}_2 \cdots \underline{p}_\ell}\big)^{-1} = \textstyle \big(1+\frac{q_1}{p_1} \cdot \big( 1+ \sum_{\ell=2}^\infty \frac{\underline{q}_2 \cdots \underline{q}_\ell}{\underline{p}_2 \cdots \underline{p}_\ell}\big) \big)^{-1}.
\end{align}
Note that for $\ell \geq 2$,
\begin{align} \label{eq:lowerboundcalc}
    a_\ell &:=  \textstyle \frac{\underline{q}_2 \cdots \underline{q}_\ell}{\underline{p}_2 \cdots \underline{p}_\ell} = \prod_{k=2}^\ell \frac{1 - C_2 k^{-1/d}}{1+ C_2 k^{-1/d}} 
    \leq \prod_{k=2}^\ell \big(1 - C_2 k^{-1/d}\big) \nonumber \\
    &\leq \textstyle  \exp\big(\!-\!\sum_{k=2}^\ell C_2 k^{-1/d} \big) \leq \exp\big( \!-\! C_2 \ell^{(d-1)/d}\big),
\end{align}
where we use that $C_2>0$ and $1-x \leq e^{-x}$ for all $x$.
Now, bounding the sum $\sum_{\ell=2}^\infty a_\ell$ by an integral and changing variables, we obtain
\begin{equation*}
    \begin{aligned}
        \textstyle \sum_{\ell=2}^\infty a_\ell &\le \textstyle C_2^{-d / (d-1)} \int_0^\infty \exp\big({-x^{(d-1)/d}}\big)dx.
    \end{aligned}
\end{equation*}
Since $(X_n)_{n \geq 0}$ stochastically lower bounds $(S_n^0)_{n \geq 0}$, the result follows.
\end{proof}

\subsection{Proof of Lemma \ref{lemma:extinctionproblargesetdb}} \label{app:shapethm1}

\begin{proof}
For $k \geq 1$, let $(X_n^k)_{n \geq 0}$ denote the simple random walk defined in the proof of Proposition \ref{prop:lowerbound} with $X_0^k = k$.
Since $(X_n^k)_{n \geq 0}$ stochastically lower bounds the embedded jump process $(S_n^{A})_{n \geq 0}$ for any $A \subseteq \mathbb{Z}^d$ with $|A|=k$, we can write
\begin{align*}
    \textstyle \sup_{A \in \overline{\cal S}, \abs{A} = k} \P(\tau_\varnothing^A < \infty) \le \P(X_m^k = 0 \text{ for some } m\ge 0).
\end{align*}
To analyze the latter probability, we first bound from below the probability that $(X_n^k)_{n \geq 0}$ ever goes below $k$.
By the generalized version of the gambler's ruin formula, we can write for $k \geq 2$,
\begin{align*}
\P\big(X_m^k \ge k \text{ for all $m \geq 0$}\big) &= \textstyle \big(\sum_{\ell=0}^\infty \prod_{n=k}^{k+\ell - 1} \frac{1-C_2 n^{-1/d}}{1 + C_2 n^{-1/d}}\big)^{-1},
\end{align*}
where $C_2 \in (0,1)$. Following a calculation analogous to~\eqref{eq:lowerboundcalc}, there is a constant $C_3 := C_2 / 2 \in (0,1/2)$ such that
$$
\textstyle a_\ell^k := \prod_{n=k}^{k+\ell - 1} \frac{1-C_2 n^{-1/d}}{1 + C_2 n^{-1/d}} \leq \exp\big(- C_2 \ell(k+\ell)^{-1/d}\big)\leq \begin{cases}
\exp\big(-C_3 \ell k^{-1/d}\big), & 0\le\ell \le k, \\
\exp\big(-C_3 \ell^{(d-1)/d}\big), & \ell > k.
\end{cases}
$$
Therefore, using $1- e^{-z} \geq z/2$ for $z\in(0,1/2)$, we obtain
$$
\textstyle\sum_{\ell=0}^k a_\ell^k \le \sum_{\ell=0}^\infty \big(e^{-C_3 k^{-1/d}}\big)^\ell = \big(1-e^{-C_3 k^{-1/d}}\big)^{-1} \leq (2/C_3) k^{1/d},$$
and
$$
\textstyle \sum_{\ell=k+1}^{\infty} a_\ell^k \leq \sum_{\ell = 0}^\infty \exp\big(-C_3 \ell^{(d-1)/d}\big)  =: M < \infty.
$$
It follows that there exists $C_4 := ((2/C_3) + M)^{-1}>0$ such that
\begin{equation}\label{escape prob}
\P\big(X_m^k \ge k \text{ for all $m \geq 0$}\big) \geq C_4 k^{-1/d}, \quad k\ge 2.
\end{equation}
Set $T_n^k := \inf\{m\ge 0 : X_m^k = n\}$. The strong Markov property and~\eqref{escape prob} imply that
\begin{align*}
&\P\big(X_m^k = 0 \text{ for some } m\ge 0\big) \\
&= \P(T_{k-1}^k < \infty, T_{k-2}^k<\infty, \ldots, T_0^k<\infty) \\
&= \textstyle \prod_{\ell=1}^k \P(T_{\ell-1}^\ell<\infty) \\
&\leq \textstyle  \P(T_0^1<\infty) \cdot  \prod_{\ell=2}^k (1 - C_4 \ell^{-1/d}) \\
&\leq \textstyle \P(T_0^1<\infty) \cdot \exp\big(-C_4 \sum_{\ell=2}^k \ell^{-1/d}\big) \\ &\leq \P(T_0^1<\infty) \cdot \exp\big(-C_4 k^{(d-1)/d}\big),
\end{align*}
which proves the lemma with $C := \P(T_0^1<\infty)>0$ and $\gamma := C_4>0$.
\end{proof}

\subsection{Proof of Lemma \ref{lemma:cannothangaroundatsmallsize}} \label{app:shapethm2}

\begin{proof}
Set $x := \varepsilon t^{d/(d+1)}$ and define $C_2 := 2C_1/(2d+1)>0$ with $C_1$ defined as in \eqref{eq:C1}.
Let $(Z_m)_{m \geq 0}$ denote the simple random walk on $\mathbb{Z}$ with $Z_0 = 1$ and transition probabilities
$$
\textstyle \P(Z_{m+1} = n+1 | Z_m = n) = 1- \P(Z_{m+1} = n-1 | Z_m=n) = \frac12(1 + C_2 x^{-1/d})
$$
for all $n\in \bZ$ and $m\ge 0$. 
If $T := \inf\{m\ge 0 : S_m^0 \notin (0,x)\}$, then $(S_{m\wedge T}^0)_{m \geq 0}$ stochastically dominates $(Z_{m\wedge T})_{m \geq 0}$.
Therefore, we can couple $(Z_m)_{m \geq 0}$ with $(S_m^0)_{m \geq 0}$ such that $Z_{m\wedge T} \le S_{m\wedge T}^0$ for all $m\ge 0$, and we can extend $Z_m$ beyond time $T$ independently.
Under this coupling, observe that $Z_m \geq x$ implies $T \leq m$. Therefore,
$$
\P(Z_m \ge x) \leq \P(S_j^0 \notin (0,x) \text{ for some $j\le m$}).
$$
Let $M_t$ be the number of jumps taken by $(|\xi_s^0|)_{s \geq 0}$ up until time $t$, and let $(N_t)_{t \geq 0}$ be a Poisson process with rate 1.  Then $(M_{t\wedge \tau_\varnothing^0})_{t \geq 0}$ stochastically dominates $(N_{t\wedge \tau_\varnothing^0})_{t \geq 0}$.  Also, observe that if $|\xi_s^0| \in (0,x)$ for all $s\le t$ and $M_t\ge t/2$, then $S_{t/2}^0 \in (0,x)$, which implies $Z_{t/2}< x$. Thus,
\begin{align} \label{jumps bound}
\P(|\xi_s^0| \in (0,x), \; s\le t)
&\le \P({N_t < t/2}) + \P({Z_{t/2} < x}) \nonumber \\
&\le e^{-t/6} + \P({Z_{t/2}<x}).
\end{align}
By Hoeffding's inequality, provided $\varepsilon>0$ is small enough so that $C_2 / (2\varepsilon^{1/d}) > 3\varepsilon/2$,
\begin{align*}
\P(Z_{t/2} < x) &\leq \P\big(|Z_{t/2} - \E Z_{t/2}| > C_2 x^{-1/d}t/2 - x\big) \\
&= \textstyle \P\big(|Z_{t/2} - \E Z_{t/2}| >(\frac{C_2}{2 \varepsilon^{1/d}} - \varepsilon) t^{d/(d+1)}\big) \\
&\leq \textstyle \P\big(|Z_{t/2} - \E Z_{t/2}| > (t/2) \cdot \varepsilon t^{-1/(d+1)}\big) \\
&\leq \textstyle 2 \exp\big(-\frac{t}{4}\cdot(\varepsilon t^{-1/(d+1)})^2 \big) \\
&= 2 \exp\big(-\gamma t^{(d-1)/(d+1)}\big),
\end{align*}
where $\gamma := \varepsilon^2/4$.
Combining with~\eqref{jumps bound}, this gives the result for large enough $t$, and $C$ can be chosen large enough so the result holds for all $t$.
\end{proof}

\subsection{Proof of Proposition \ref{prop:shapethmdb}} \label{app:shapethm3}

\begin{proof}
To prove a shape theorem for the death-birth model, we need to verify a slight weakening of the conditions given in Durrett and Griffeath \cite{DurrettGriffeath82}.
First, define
$$
K_t :=\{y\in\mathbb{R}^d: \exists \, x\in\xi_t^0\mbox{ s.t. } ||x-y||_{\infty}\leq 1/2\}, \quad t \geq 0.
$$
The conditions to be verified are that there exist constants $\gamma, c, C, p \in (0,\infty)$ such that
\begin{align} 
& \P(t < \tau_\varnothing^0 < \infty) \le C e^{-\gamma t^p}, \quad t\ge 0, \label{eq:extinctTimeCond} \\
& \P(x\notin K_t, \tau_\varnothing^0 = \infty) \le Ce^{-\gamma t^p}, \quad t\ge 0, \; ||x||_{\infty}< ct. \label{eq:growthRateCond}
\end{align}
This is a weakening of the conditions in \cite{DurrettGriffeath82} since our conditions allow for $p\neq 1$, in particular $p \in (0,1)$.
With this weakening, condition \eqref{eq:growthRateCond} follows immediately from Proposition \ref{prop:linexpball} with $p=1/2$.
Condition \eqref{eq:extinctTimeCond} follows from Lemmas \ref{lemma:extinctionproblargesetdb} and \ref{lemma:cannothangaroundatsmallsize}, as we now show.
Define $V_k := \inf\{t \geq 0: \abs{\xi_t^0} = k\}$ for $k \geq 1$ and note that by the strong Markov property,
\begin{align*}
\P(t<\tau_\varnothing^0<\infty, V_k < t) &\le \textstyle \sum_{A\subseteq \bZ^d, \abs{A} = k} \P(\tau_\varnothing^0 < \infty, V_k<t, \xi_{V_k}^0 = A) \\
&\le \textstyle \sum_{A\subseteq \bZ^d, \abs{A} = k} \P(\tau_\varnothing^0 < \infty, \xi_{V_k}^0 = A  |  V_k < \infty) \\
&= \textstyle \sum_{A\subseteq \bZ^d, \abs{A} = k} \P(\tau_\varnothing^A < \infty) \P(\xi_{V_k}^0 = A  |  V_k < \infty) \\
&\le \textstyle \sup_{A\subseteq \bZ^d, \abs{A} = k} \P(\tau_\varnothing^A < \infty).
\end{align*}
Then note that
\begin{align*}
    \P(t<\tau_\varnothing^0<\infty, V_k>t) \leq \P(|\xi_s^0| \in (0,k), \; s\le t).
\end{align*}
By selecting $\varepsilon>0$ sufficiently small and setting $k = \lfloor \varepsilon t^{d/(d+1)}\rfloor$, condition \eqref{eq:extinctTimeCond} now follows from Lemmas \ref{lemma:extinctionproblargesetdb} and \ref{lemma:cannothangaroundatsmallsize} with $p=(d-1)/(d+1)$.
\end{proof}

\noindent {\bf Acknowledgments.}
This manuscript is dedicated to the memory of Ching-Shan Chou, our friend and colleague.
We would like to acknowledge Rick Durrett for his insights on the dual process for the death-birth model during discussions at the Mathematical Biosciences Institute's thematic program on cancer evolution (2014). \\

\noindent {\bf Funding.}
JF was supported in part by NSF grants DMS-1349724 and DMS-2052465.
EBG and KL were supported in part by NSF grant CMMI-1552764.
DS was supported in part by NSF grant CCF-1740761.
JF and KL were supported in part by the U.S.-Norway Fulbright Foundation and the Research Council of Norway R\&D Grant 309273.
EBG was supported in part by the Norwegian Centennial Chair grant and the Doctoral Dissertation Fellowship from the University of Minnesota.
\\

\noindent {\bf Conflict of Interest Statement.}
On behalf of all authors, the corresponding author states that there is no conflict of interest. 
 \bibliographystyle{unsrt}
\bibliography{citations}

\end{document}